\documentclass[journal,onecolumn]{IEEEtran}

\usepackage[english]{babel}
\usepackage{booktabs}

\usepackage{ifpdf}

\usepackage{cite} 
\usepackage{url}
\usepackage{hyperref}

\ifCLASSINFOpdf
	\usepackage[pdftex]{graphicx}
	\graphicspath{{./figures/}}
\else
	\usepackage[dvips]{graphicx}
	\graphicspath{./figures/}
\fi
\usepackage{color}

\usepackage{pgf, tikz, pgfplots}
\usetikzlibrary{shapes, arrows, automata}
\usetikzlibrary{calc,hobby,decorations}

\usepackage[cmex10]{amsmath}
\usepackage{amsfonts, amssymb, amsthm}
\usepackage{mathrsfs}

\usepackage{algorithm} 
\usepackage{algorithmic}
\usepackage{enumerate}
\usepackage{multirow}
\usepackage{rotating}
\usepackage{subcaption}
\usepackage{needspace}

\input{mySymbol.sty}

\definecolor{penndarkestblue}{cmyk}{1,0.74,0,0.77}
\definecolor{penndarkerblue}{cmyk}{1,0.74,0,0.70}
\definecolor{pennblue}{cmyk}{0.99,0.66,0,0.57} 
\definecolor{pennlighterblue}{cmyk}{0.98,0.44,0,0.35}
\definecolor{pennlightestblue}{cmyk}{0.38,0.17,0,0.17} 

\definecolor{penndarkestred}{cmyk}{0,1,0.89,0.66}
\definecolor{penndarkerred}{cmyk}{0,1,0.88,0.55}
\definecolor{pennred}{cmyk}{0,1,0.83,0.42} 
\definecolor{pennlighterred}{cmyk}{0,1,0.6,0.24}
\definecolor{pennlightestred}{cmyk}{0,0.43,0.26,0.12} 

\definecolor{penndarkestgreen}{cmyk}{1,0,1,0.68}
\definecolor{penndarkergreen}{cmyk}{1,0,1,0.57}
\definecolor{penngreen}{cmyk}{1,0,1,0.44} 
\definecolor{pennlightergreen}{cmyk}{1,0,1,0.25}
\definecolor{pennlightestgreen}{cmyk}{0.43,0,0.43,0.13}

\definecolor{penndarkestorange}{cmyk}{0,0.65,1,0.49}
\definecolor{penndarkerorange}{cmyk}{0,0.65,1,0.33}
\definecolor{pennorange}{cmyk}{0,0.54,1,0.24} 
\definecolor{pennlighterorange}{cmyk}{0,0.32,1,0.13}
\definecolor{pennlightestorange}{cmyk}{0,0.15,0.46,0.06}
	
\definecolor{penndarkestpurple}{cmyk}{0,1,0.11,0.86}
\definecolor{penndarkerpurple}{cmyk}{0,1,0.13,0.82}
\definecolor{pennpurple}{cmyk}{0,1,0.11,0.71} 
\definecolor{pennlighterpurple}{cmyk}{0,1,0.05,0.46}
\definecolor{pennlightestpurple}{cmyk}{0,0.35,0.02,0.23}
	
\definecolor{pennyellow}{cmyk}{0,0.20,1,0.05} 
\definecolor{pennlightgray1}{cmyk}{0,0,0,0.05}
\definecolor{pennlightgray3}{cmyk}{0.01,0.01,0,0.18}
\definecolor{pennmediumgray1}{cmyk}{0.04,0.03,0,0.31}
\definecolor{pennmediumgray4}{cmyk}{0.08,0.06,0,0.54}
\definecolor{penndarkgray2}{cmyk}{0.09,0.07,0,0.71}
\definecolor{penndarkgray4}{cmyk}{0.1,0.1,0,0.92}

\def\SO3{\mathrm{SO(3)}}

\newtheorem{assumption}{\hspace{0pt}\bf Assumption \hspace{-0.15cm}}

\newtheorem{proposition}{\hspace{0pt}\bf Proposition}

\newtheorem{theorem}{\hspace{0pt}\bf Theorem}
\newtheorem{corollary}{\hspace{0pt}\bf Corollary}

\newtheorem{remark}{\hspace{0pt}\bf Remark}

\begin{document}

\title{Balancing Rates and Variance via Adaptive Batch-Size for Stochastic Optimization Problems
}
\author{Zhan Gao$^\star$, Alec Koppel$^\dagger$, and Alejandro Ribeiro$^\star$ 
\thanks{The work in this paper was supported by ARL DCIST CRA W911NF-17-2-0181 and NSF HDR TRIPODS 1934960. Preliminary results appear in ICASSP 2020 conference \cite{cGao2020}. $^\star$Department of Electrical and Systems Engineering, University of Pennsylvania, Philadelphia, PA (Email: \{gaozhan,aribeiro\}@seas.upenn.edu). $^\dagger$Computational and Information Sciences Directorate, U.S. Army Research Laboratory, 2800 Powder Mill Rd., Adelphi, MD 20783 (alec.e.koppel.civ@mail.mil).}

}

\maketitle

\begin{abstract}
Stochastic gradient descent is a canonical tool for addressing stochastic optimization problems, and forms the bedrock of modern machine learning and statistics. In this work, we seek to balance the fact that attenuating step-size is required for exact asymptotic convergence with the fact that constant step-size learns faster in finite time up to an error. To do so, rather than fixing the mini-batch and the step-size at the outset, we propose a strategy to allow parameters to evolve adaptively. Specifically, the batch-size is set to be a piecewise-constant increasing sequence where the increase occurs when a suitable error criterion is satisfied. Moreover, the step-size is selected as that which yields the fastest convergence. The overall algorithm, two scale adaptive (TSA) scheme, is developed for both convex and non-convex stochastic optimization problems. It inherits the exact asymptotic convergence of stochastic gradient method. More importantly, the optimal error decreasing rate is achieved theoretically, as well as an overall reduction in computational cost. Experimentally, we observe that TSA attains a favorable tradeoff relative to standard SGD that fixes the mini-batch and the step-size, or simply allowing one to increase or decrease respectively. 

%

\end{abstract}

\begin{IEEEkeywords}
Stochastic optimization, stochastic gradient descent, adaptive batch-size, optimal step-size
\end{IEEEkeywords}

\IEEEpeerreviewmaketitle


\section{Introduction} \label{sec:intro}

Many machine learning \cite{Kingma2015}, control \cite{Pham2009}, and signal processing tasks \cite{Pereyra2015, solo1994adaptive, sayed2003fundamentals} may be formulated as finding the minimizer of an expected cost parameterized by a random variable that encodes data. In particular, communications channel estimation \cite{ribeiro2010ergodic}, learning model mismatch of a dynamical system \cite{koppel2017proximity}, and training modern vision systems \cite{bray20043d}, hinge upon solving stochastic optimization problems. Stochastic gradient descent (SGD) is a widely used approach for these problems \cite{Ruszczynski1983, mathews1993stochastic, erdogmus2006linear, zhang2004solving, bottou2010large, kolodziej2008constrained, amiri2020machine}, but practically its parameter tuning can be difficult. This is because attenuating step-size is required for exact solutions, which slows learning to a stand still. Moreover, mini-batch size is typically fixed at the outset of training, which reduces variance by a set amount \cite{bouleau1993numerical}. In this work, we balance the choice of step and mini-batch sizes to learn quickly with ever-decreasing variance by allowing batch-size to enlarge adaptively.

To contextualize our approach, we begin by noting that gradient iteration yields an effective approach for both convex and nonconvex minimization problems \cite{Cauchy1847, Nesterov2004, chartrand2007exact, ruder2016overview, zeng2018nonconvex}. However, when the objective takes the form of an expected value, evaluating the gradient direction is intractable. To surmount this issue, SGD descends along stochastic gradients in lieu of true gradients, i.e., gradients evaluated at a single (or possibly multiple) sample point(s) \cite{Robbins1951, amari1993backpropagation, song2013stochastic, mokhtari2017large}. Its simplicity and  theoretical guarantees make it appealing, but its practical usage requires parameter tuning that can be unintuitive.

Specifically, it's well-understood that attenuating step-size is theoretically required for exact convergence, at the cost of reducing the learning speed to null as time progresses \cite{Wardi1989, gardner1984learning}. Experimentally, constant step-size vastly improves performance, but only attains approximate convergence \cite{Nedic2001,dieuleveut2017bridging}. The choice of step-size in terms of learning rate typically depends on the Lipschitz modulus of continuity, which then affects the asymptotic mean square error of the parameter estimate \cite{mokhtari2016class, tan2016barzilai, li2019convergence}. Hypothetically, one would like to preserve fast rates while ever-tightening the radius of convergence, to obtain exact convergence in the limit.
 
To do so, we shift focus to mini-batching. Mini-batching is a procedure where the stochastic gradient is averaged over multiple samples per iteration. Under constant learning rates, its practical effect is to reduce the variance of stochastic approximation error, which tightens asymptotic  convergence \cite{Li2014, Konecny2015, yuan2016stochastic}. Intriguingly, it has recently been established that when the batch-size grows geometrically with the iteration index, SGD obtains exact solutions in the limit even under fixed step-size \cite{Byrd2012, bottou2018optimization, balles2017coupling}. With the drawback of large sample complexity, this fact then motivates allowing the batch-size to grow as slowly as possible, while maintaining the exact convergence and an optimal learning rate determined by problem smoothness. Doing so is exactly the proposed strategy of this work, which we establish yields an overall reduction in sample complexity relative to standard approaches.

By co-considering the batch-size and the step-size simultaneously, this paper proposes an optimal parameter selection algorithm, the two scale adaptive (TSA) scheme, for both convex and non-convex stochastic optimization problems. Two sub-schemes, the post TSA and the prior TSA, are developed with emphasis on sample complexity and learning rate during the tradeoff, respectively. In Section \ref{sec:problem}, we clarify details of the problem definition and the standard SGD method. In Section \ref{sec:algorithm}, we provide preliminary analysis for SGD on convex problems, based on which we develop the convex TSA scheme that outperforms standard SGDs. In Section \ref{sec:nonalgorithm}, we propose a variant of TSA for non-convex problems. In Section \ref{sec:convergence}, we establish our main results that characterize the convergence, rate and sample complexity of TSA. We show an overall reduction in sample complexity of TSA under fast learning rates, equipped with the exact convergence.

\begin{table*}[]

\normalsize

\centering

\renewcommand\tabcolsep{5pt}

\caption{Two scale adaptive (TSA) scheme for convex and non-convex problems. TSA performs the inner time-scale and outer time-scale recursively with $t$ the number of inner/outer time-scales and $k$ the number of iterations. The $t$-th inner time-scale runs SGD with the optimal step-size $\alpha_t = 1/L$ and batch-size $n_t$ until the largest iteration before $Q_1^t$ drops below $Q_2^t$. The $t$-th outer time-scale follows to increase the batch-size by the addition $n_{t+1}=n_t + \beta$ or the multiplication $n_{t+1} = m n_t$. Here $\ell, L$ and $w$ are problem parameters, $\beta$ and $m$ the additive and multiplicative parameters, $K_t$ the duration of $t$-th inner time-scale, and $A = \log_{1-\ell/L}(1/2m)$, $B=\log_{1-\ell/L}(1/m)$ and $D=F(\bbx_0)\!-\!F(\bbx^*)$ are concise constants for presentation. }

\label{Tab03}

\begin{tabular}{ccccccccccccc}

\toprule

\multirow{1}{*}{Method} & \multicolumn{1}{c}{$Q_1^t$} & \multicolumn{1}{c}{$Q_2^t$} & \multicolumn{1}{c}{Exact convergence} & \multicolumn{1}{c}{Rate (addition)} & \multicolumn{1}{c}{Rate (multiplication)} \\

\midrule

Convex post TSA    & $2^{t}\!\left(1\!-\!\frac{\ell}{L} \!\right)^k \!D$         & $\frac{w}{2 n_t\ell}$        & $\checkmark$  & $\ccalO \left(2^t\left(1-\frac{\ell}{L}\right)^k\right)$ &   $\ccalO\left(m^{-\frac{k}{A}}\right)$  \\[5pt]

Convex prior TSA                                     & $\left( 1 \!-\! \frac{\ell}{L} \right)^{k}\! D$                    & $\frac{w}{2 n_t\ell}$   & $\checkmark$           & $\ccalO\left(t\left(1-\frac{\ell}{L}\right)^k \right)$  &      $\ccalO\left(\!\frac{k}{B}\left(1-\frac{\ell}{L}\right)^k\!\right)$       \\[5pt]

Non-convex post TSA                                    & $\frac{1}{k}\! \left( 2LD \!+\!\sum_{i=0}^{t\!-\!1}\!\frac{K_i w}{ n_i} \right)$                    &                   $\frac{w}{n_t}$     & $\checkmark$           & $\backslash$  &  $\backslash$           \\[5pt]

Non-convex prior TSA                                      & $\frac{1}{k} 2LD $    &                   $\frac{w}{n_t}$  &  $\checkmark$ &$\ccalO\left(\frac{\log k}{k}\right)$&  $\ccalO\left(\frac{\log k}{k}\right)$  \\[5pt]

\bottomrule

\end{tabular}

\end{table*}

We perform experiments on the supervised classification problem of hand-written digits with the linear model (convex) and the convolutional neural network (non-convex) to corroborate our model and theory in Section \ref{sec:sims}. Lastly, we conclude the paper in Section \ref{sec:conclusion}. Overall, by evolving parameters with the TSA scheme, we minimize both asymptotic bias and sample path variance of the learning procedure, and dispel the need of presetting SGD parameters at the outset, where the latter deteriorates performance in practice.

\section{Problem Formulation} \label{sec:problem}

Denote by $\bbx \in \mathbb{R}^p$ a decision vector to be determined, e.g., the parameters of a statistical model, $\bbxi \in \Omega$ a random variable with the probability distribution $\bbp$, and $f(\bbx, \bbxi)$ an objective function whose average minimizer we seek to compute. For instance, realizations $\bbxi_i$ of $\bbxi$ are training examples $(\bbz_i, \bby_i)$ in supervised learning, and $f(\bbx, \bbxi)$ represents the model fitness of $\bbx$ at data points $(\bbz_i, \bby_i)$. We consider the stochastic optimization problem defined as
\begin{equation}\label{eq:main_prob}
\begin{split}
\min_{\bbx} F(\bbx) = \min_{\bbx} \int_{\Omega} f(\bbx,\bbxi) p(d\bbxi) = \min_{\bbx} \mathbb{E} [f(\bbx,\bbxi)]
\end{split}
\end{equation}
%
where the population cost $F(\bbx)$ can be either convex (Sec. \ref{sec:algorithm}) or non-convex (Sec. \ref{sec:nonalgorithm}). Typically, the distribution $\bbp$ of $\bbxi$ is unknown, and therefore $F(\bbx)$ in \eqref{eq:main_prob} is not computable. This issue usually motivates drawing $N$ samples from $\bbp$ and solving the corresponding empirical risk minimization (ERM) problem \cite{chaudhuri2011differentially}
\begin{equation}\label{eq:empirical}
\min_{\bbx} F(\bbx) = \min_{\bbx} \frac{1}{N}\sum_{i=1}^N f(\bbx,\bbxi_i) \; 
\end{equation}
where, for instance, $ \{ f(\bbx,\bbxi_i)\}$ may represent the hinge or logistic loss at data points $\{\bbxi_i\}$ of logistic regression or support vector machine classification \cite{King2000}, respectively. 
%

We set our focus on the general population problem \eqref{eq:main_prob} in this paper. Define the simplifying notation $f_i(\bbx) := f(\bbx,\bbxi_i)$ and denote by $\nabla f_\ccalS(\bbx)$ the average gradient of a sample set $\ccalS=\{\bbxi_1,\dots,\bbxi_n\}$ as
\begin{equation}\label{eq:mini_batch|}
\nabla f_S(\bbx)=\frac{1}{n}\sum_{\bbxi_i\in S}\nabla f_i(\bbx).
%
\end{equation}
The preceding expression is commonly referred to as the mini-batch gradient. To propose our algorithm for \eqref{eq:main_prob} , we require the following assumptions that are commonly satisfied in practical optimization problems \cite{Byrd2012, bottou2018optimization}.
\begin{assumption} \label{asm1}
The gradient of expected objective function $\nabla F(\bbx)$ is Lipschitz continuous, i.e., for all $ \bbx, \bby \in \mathbb{R}^p $, there exists a constant $L$ such that
\begin{equation}
\begin{split}
\| \nabla F(\bbx) - \nabla F(\bby) \|_2 \le L \| \bbx - \bby \|_2,
\end{split}
\end{equation}
where $\| \cdot \|_2$ is the ($2$-) norm.
\end{assumption}
\begin{assumption}\label{as2}
The objective functions $f(\bbx, \bbxi)$ are differentiable, and for any $\bbx \in \mathbb{R}^p$, there exists a constant $w$ such that
\begin{equation}
\begin{split}
\| {\rm var}[\nabla f(\bbx, \bbxi)] \|_1 \le w,
\end{split}
\end{equation}
where $\| \cdot \|_1$ is the ($1$-) norm and ${\rm var}[\nabla f(\bbx, \bbxi)]$ is the variance vector, each component of which is the variance of corresponding component of vector $\nabla f(\bbx, \bbxi)$.
\end{assumption}

Mini-batch stochastic gradient descent (SGD) is a standard approach to solve \eqref{eq:main_prob}. Specifically, it estimates the gradient unbiasedly with a random mini-batch gradient at each iteration, based on which it updates the decision vector $\bbx$ as 
\begin{equation}\label{eq:sgd}
\begin{split}
\bbx_{k+1} = \bbx_k - \alpha_k \nabla f_{S_k}(\bbx_k).
\end{split}
\end{equation}
Here, $\alpha_k$ is the step-size and $S_k$ is the mini-batch selected at $k$-th iteration. In this work, we are concerned with optimal selections of batch-size $n_k:=\vert S_k \vert$ and step-size $\alpha_k$ in according to their ability to maximize the convergence progress of each sample, and thus not needlessly waste data. Next we shift gears to doing so in terms of mean convergence for convex and non-convex problems, respectively.



\section{Adaptive Batching in Convex Problems} \label{sec:algorithm}

We begin by considering convex problems. Specifically, in addition to Assumptions  \ref{asm1}-\ref{as2}, we impose the condition that the population cost $F(\bbx)$ is strongly convex as stated next.

\begin{assumption} \label{as3}
The objective functions $\{ f_i(\bbx) \}$ are differentiable, and the expected objective $F(\bbx)$ is strongly convex, i.e., for all $\bbx, \bby \in \mathbb{R}^p$, there exists a constant $\ell$ such that
\begin{equation}
\begin{split}
F(\bbx) \ge F(\bby)+ \nabla F(\bby)^\top (\bbx-\bby) + \frac{\ell}{2}\| \bbx-\bby \|^2_2.
\end{split}
\end{equation}
\end{assumption}

Under Assumptions  \ref{asm1}-\ref{as3}, we may characterize the convergence rate of \eqref{eq:sgd} in expectation as a function of batch-size $n_k$ and step-size $\alpha_k$, which forms the foundation of our parameter scheduling scheme.
\begin{proposition}\label{prop1}
Under Assumptions \ref{asm1}-\ref{as3}, the expected sub-optimality sequence of SGD \eqref{eq:sgd} with batch-size at iteration $k$ as $\vert S_k \vert = n_k$, satisfies
\begin{equation}\label{eq:decrement}
\begin{split}
 \mathbb{E}\left[ F(\bbx_{k+1})-F(\bbx^*)\right] &\le  \left( \prod_{i=0}^k r(\alpha_i) \right) \left( F(\bbx_{0})-F(\bbx^*) \right) + \sum_{i=0}^k \left( \frac{\alpha_i^2 L w}{2n_i} \prod_{j=i+1}^k r(\alpha_j) \right)
\end{split}
\end{equation}
where $r(\alpha)=1-2\alpha \ell + L\ell\alpha^2$ is the modulus of contraction, and the expectation $\mathbb{E}[\cdot]$ is over the unknown distribution of $\bbxi$ whose samples $\ccalS_k = \{\bbxi_j \}_{j=1}^{n_k}$ are observed at iteration $k$.
\end{proposition}

\begin{proof}
See Appendix \ref{pr:prop1}.
\end{proof}

Proposition \ref{prop1} establishes the dependence of convergence rate on problem constants, step-size and batch-size selections.
 In particular, we observe an inverse dependence on the batch-size with respect to the later term, and a suggested optimal value on the step-size in order to make the first term decay quickly. This relationship may be employed to discern ways of reducing the overall computational complexity of \eqref{eq:sgd} required to attain the optimal solution of \eqref{eq:main_prob} with fast rates. 
 
For special cases with constant $\alpha_k = \alpha$ and $n_k = n$, Proposition \ref{prop1} reduces to the following corollary.

\begin{corollary} \label{coro1}
With the same settings as Proposition \ref{prop1}, assume that the step-size and the batch-size are constants as $\alpha_k = \alpha$ and $n_k = n$. Then, it holds that
\begin{align} \label{eq:coro1main}
&\mathbb{E}\left[ F(\bbx_{k+1})-F(\bbx^*)\right] \le \underbrace{r(\alpha)^{k+1} \left( F(\bbx_{0})-F(\bbx^*)\right)}_{:=Q_1}+ \underbrace{\frac{\alpha L w}{2n(2\ell-L\ell\alpha)}}_{:=Q_2}
\end{align}
with $r(\alpha) = 1-2\alpha \ell + L\ell\alpha^2$.
\end{corollary}

\begin{proof}
See Appendix \ref{pr:coro11}.
\end{proof}

The bound in \eqref{eq:coro1main} jointly depends on the step-size $\alpha$ and batch-size $n$, and consists of two terms. $Q_1$, the convergence rate term, decreases with iteration $k$ provided $r(\alpha) <1$. The error neighborhood term $Q_2$ typically determines the limiting radius of convergence for constant step-size, and is associated with the variance of stochastic approximation (sub-sampling) error. Our algorithm, the two scale adaptive (TSA) scheme, is then proposed by exploiting the structure of $Q_1$ and $Q_2$.

\subsection{Two Scale Adaptive Scheme}

Observe from \eqref{eq:coro1main} that $Q_1$ decreases monotonically while $Q_2$ stays constant as iteration $k$ increases. Once $Q_1$ decays to be less than $Q_2$, $\bbx_k$ cannot converge to a tighter neighborhood than that determined by $Q_2$. Therefore, in order to continue tightening the radius of convergence, one must reduce $Q_2$ by either decreasing the step-size $\alpha$ or increasing the batch-size $n$. The TSA scheme gives a strategy about when and how to make this change. The algorithm is divided into two scales: the inner time-scale performs SGD with constant step-size $\alpha$ and batch-size $n$, and the outer time-scale tunes algorithm parameters $\alpha$, $n$ to tighten the radius of convergence. Two sub-schemes: the post TSA and the prior TSA, are proposed based on different stop criterions of the inner time-scale. Details are formally introduced below.  

\emph{Initialization}: With initial step-size $\alpha = \alpha_0$ and batch-size $n = n_0$, the corresponding $Q_1^0$ and $Q_2^0$ are defined as \eqref{eq:coro1main}. Without loss of generality, assume $F(\bbx_{0})-F(\bbx^*)$ is large such that $Q_1^0 \ge Q_2^0$ initially. Then $Q_1^0$ dominates, while $Q_2^0$ is relatively small. Note that the decreasing rate of $Q_1^0$ is $ r(\alpha_0)$, which is a quadratic function of step-size $\alpha_0$. To ensure an optimal decrease, $\alpha_0=1/L$ is selected to achieve its minimal value $1-\ell/L$. The corresponding expression for $Q_2^0$ is
\begin{equation}
\begin{split}
Q_2^0 = \frac{\frac{1}{L} L w}{2n_0(2\ell-L\ell\frac{1}{L})}= \frac{w}{2  n_0\ell}.
\end{split}
\end{equation}
Fixing step-size $\alpha=1/L$ ensures the fastest decrease of $Q_1$, during which we propose evolving the batch-size $n$ to tighten the error neighborhood $Q_2$. Proceeding from this initialization, we shift to discuss progressively enlarging the batch-size.

(1-A) \emph{Inner time-scale (post)}. At $t$-th inner time-scale, let $\alpha_{t} = 1/L$ and $n_{t}$ be the current step-size and batch-size, and $K$ the beginning number of iteration. Follow the SGD method with constant $\alpha_{t}$ and $n_{t}$, and define $Q_1^t$ and $Q_2^t$ as
\begin{align}\label{eq:error_conditions}
&Q_1^t = \left( 1 - \frac{\ell}{L} \right)^{k_t} \mathbb{E}\left[ F(\bbx_{K})-F(\bbx^*) \right],\\
\label{eq:error_conditions2}&Q_2^t = \frac{\alpha_t L w}{2n_t(2\ell-L\ell\alpha_t)} = \frac{w}{2 n_t\ell} 
\end{align} 
with $k_t >0$ the passed number of iterations at $t$-th inner time-scale. As $k_t$ increases, $Q_1^t$ decreases multiplied by $1-\ell/L$, while $Q_2^t$ stays constant. Then there exists $\tilde{K}_t$ such that 
\begin{equation} \label{eq:innercond}
\begin{split}
\tilde{K}_t =  \max_{k_t} \{ Q_1^t \ge Q_2^t\}. 
\end{split}
\end{equation}
$\tilde{K}_t$ is the largest iteration before $Q_1^t$ drops below $Q_2^t$ and named as the duration of $t$-th inner time-scale. 

Unfortunately, since the sub-optimality at $K$-th iteration $\mathbb{E}\left[ F(\bbx_{K})-F(\bbx^*) \right]$ is unknown, the condition \eqref{eq:innercond} cannot be used directly. Therefore, we propose the modified rule
%
\begin{equation} \label{eq:innercond05}
\begin{split}
K_t \!=  \!\max_{k_t}\! \left\{\! 2^{t}\!\!\left(1-\frac{\ell}{L} \!\right)^{\!\!K + k_{t}} \!\!\!\!\!\!\!\ \left( F(\bbx_0)-F(\bbx^*)\right) 
 \ge \frac{w}{2 n_t\ell} \right\}\!. 
\end{split}
\end{equation}
Let $\{ \tilde{K}_0, \tilde{K}_1,..., \tilde{K}_{t-1} \}$ be durations of previous inner time-scales such that $K=\sum_{i = 0}^{t-1} \tilde{K}_i$. The modified condition \eqref{eq:innercond05} is derived by utilizing the result \eqref{eq:coro1main} in Corollary \ref{coro1}
 \begin{equation} \label{eq:innercond051}
\begin{split}
\mathbb{E}\left[ F(\bbx_{K})-F(\bbx^*) \right] &= \mathbb{E}\left[ F(\bbx_{\sum_{i = 0}^{t-1} \tilde{K}_i})-F(\bbx^*) \right] \le \left( 1 - \frac{\ell}{L} \right)^{\tilde{K}_{t-1}} \mathbb{E}\left[ F(\bbx_{\sum_{i = 0}^{t-2} \tilde{K}_i})-F(\bbx^*) \right] + Q_2^{t-1}\\
&\le 2 \left( 1 - \frac{\ell}{L} \right)^{\tilde{K}_{t-1}} \mathbb{E}\left[ F(\bbx_{\sum_{i = 0}^{t-2} \tilde{K}_i})-F(\bbx^*) \right]
\end{split}
\end{equation}
where the last inequality follows from the definition of $\tilde{K}_{t-1}$ given in \eqref{eq:innercond}. Recursively applying this property, we obtain that $Q_1^t$ is bounded by
\begin{equation} \label{eq:conalg1}
\begin{split}
&Q_1^t\le 2^{t}\left(1-\frac{\ell}{L} \right)^{K + k_{t}} \left( F(\bbx_0)-F(\bbx^*)\right).
\end{split}
\end{equation}
Then, \eqref{eq:innercond05} is derived from using the fact that the preceding expression provides a computable condition for which $Q_1^t \geq Q_2^t$, provided an initial estimate on the sub-optimality is available. Whenever it is not, a large constant can replace it. Specifically, in practice, by assuming $\vert F(\bbx) \vert$ is bounded, $F(\bbx_0)-F(\bbx^*)$ can be approximated by $\max_{\bbx, \bby} \vert F(\bbx) - F(\bby) \vert$ initially. 
%
Once the iteration $k_t$ reaches $K_t$, we move forward to the $t$-th outer time-scale to augment the batch-size.

(1-B) \emph{Inner time-scale (prior)}. At $t$-th inner time-scale, let $\alpha_{t} = 1/L$, $n_{t}$ and $K$ be the current step-size, batch-size and the beginning number of iteration. Perform the SGD with constant $\alpha_{t}$ and $n_{t}$, and define $Q_1^t$ and $Q_2^t$ as
\begin{align}\label{eq:error_conditions1}
&Q_1^t = \left( 1 - \frac{\ell}{L} \right)^{K + k_{t}} \left( F(\bbx_{0})-F(\bbx^*)\right),\\
&Q_2^t = \frac{\alpha_t L w}{2n_t(2\ell-L\ell\alpha_t)} = \frac{w}{2 n_t\ell} 
\end{align} 
with $k_t >0$ the passed number of iterations at $t$-th inner time-scale. With the increasing of $k_t$, $Q_1^t$ decreases multiplied by $1-\ell/L$ recursively, while $Q_2^t$ stays constant. Define the duration of $t$-th inner time-scale $K_t$ as
\begin{equation} \label{eq:innercond1}
\begin{split}
K_t =  \max_{k_t} \left\{ Q_1^t \ge Q_2^t \right\}.
\end{split}
\end{equation}
Once $k_t$ reaches $K_t$, we stop the $t$-th inner time-scale and augment the batch-size, as we detail next.

\begin{remark}\normalfont
The difference between the post (I-A) and the prior (I-B) sub-schemes lies on that $Q_1^t$ of the post depends on the sub-optimality $\mathbb{E}\left[ F(\bbx_{K})-F(\bbx^*) \right]$ at iteration $K$, while $Q_1^t$ of the prior directly utilizes the initial sub-optimality $F(\bbx_{0})-F(\bbx^*)$ without accounting for accumulated error neighborhood terms $\{ Q_2^0, \ldots, Q_2^{t-1} \}$ as demonstrated in \eqref{eq:innercond051}. As such, the prior will increase the batch-size faster  with a faster convergence rate, but also require more sample complexity compared with the post. Two sub-schemes imply two adaptive strategies that balance rates and variance with different emphasis.   
\end{remark}

(2) \emph{Outer time-scale}. At $t$-th outer time-scale, we evolve parameters to reduce the error neighborhood $Q_2^t$, which can be realized by either the decreasing of $\alpha_t$ or the increasing of $n_t$. The former slows down the convergence rate $r(\alpha_t)$, while the latter increases the sample complexity. The tradeoff between these two factors needs to be judiciously balanced.

Once the condition \eqref{eq:innercond05} (the post) or \eqref{eq:innercond1} (the prior) is satisfied, we increment the current batch-size $n_t$ to $n_{t+1}$ in one of two possible ways: addition and multiplication.
\begin{align}
\label{eq:add} n_{t+1} = n_t + \beta, \quad \beta \ge 1,\\
\label{eq:multi} n_{t+1} = m n_t, \quad m>1
\end{align}
where $\beta$ and $m$ are additive and multiplicative integer parameters, respectively. Though the selection of \eqref{eq:add} and \eqref{eq:multi} is not the key point in TSA, it is an available tradeoff that can be tuned in practice to help improve performance. So far $t$-th outer time-scale has been completed, and $(t+1)$-th inner time-scale follows recursively. 

Note that $\beta$ or $m$ is selected appropriately to ensure $K_t$ in \eqref{eq:innercond05} or \eqref{eq:innercond1} larger than zero. The post TSA and the prior TSA share a similar process but with different stop criterions in the inner time-scale. Together, the TSA scheme for convex problems is summarized as in Algorithm \ref{alg:TSAconve} with $t$ the number of inner/outer time-scales and $k$ the number of iterations. 
{\linespread{1}
\begin{algorithm}[t] \begin{algorithmic}[1]
\STATE \textbf{Input:} objective functions $\{ f(\bbx,\bbxi) \}$, decision vector $\bbx_0$
\STATE Set step-size $\alpha = 1/L$; sample-size $\vert S_0 \vert=n_0$ and $t=0$
\STATE Compute $Q_1^t =  F(\bbx_0)-F(\bbx^*)$, $Q_2^t = \frac{w}{2\ell n_0}$
\FOR [main loop]{$k = 0,1,2...$}
      \STATE Update the decision vector $\bbx_{k+1} = \bbx_k - \alpha \nabla f_{S_k}(\bbx_k)$
	  \STATE Compute $Q_1^t = \left(1-\frac{\ell}{L} \right) Q_1^t$ 
	  \IF {$\left(1-\frac{\ell}{L} \right)Q_1^t \le Q_2^t$} 
	      \STATE Update batch-size $n_{t+1} = n_t + \beta$ or $n_{t+1} = m n_t$
	      \STATE Post: Update $Q_1^{t+1} = 2 Q_1^t$, $Q_2^{t+1} = \frac{n_{t} Q_2^t}{n_{t+1}}$, $t=t+1$
	      \STATE Prior: Update $Q_2^{t+1} = \frac{n_{t} Q_2^t}{n_{t+1}}$, $t=t+1$
      \ENDIF  
\ENDFOR
\end{algorithmic}
\caption{Two Scale Adaptive Scheme (Convex Case)}\label{alg:TSAconve} \end{algorithm}}

\section{Adaptive Batching in Non-Convex Problems} \label{sec:nonalgorithm}

In this section, we present a modification of the TSA scheme that applies non-convex problems, i.e., $F(\bbx)$ in \eqref{eq:main_prob} does not satisfy Assumption \ref{as3}. Differently from the convex case, $F(\bbx) - F(\bbx^*)$ is no longer appropriate as the convergence criterion due to the lack of a decrement property [cf. \eqref{eq:decrement}] holding. Instead, we exploit the fact that the quantity $\| \nabla_\bbx F(\bbx) \|^2_2$ has a similar magnitude to $F(\bbx) - F(\bbx^*)$, and thus may replace the sub-optimality as a convergence criterion in the non-convex regime \cite{Kingma2015}. Thus, we start by characterizing the convergence rate of SGD \eqref{eq:sgd} in expectation in terms of this alternative criterion which obviates the need for convexity. This forms the basis upon which we develop the modified TSA scheme.

\begin{proposition}\label{prop2}
Under Assumptions \ref{asm1}-\ref{as2}, the SGD sequence \eqref{eq:sgd} with step-size and batch-size at iteration $k$ as $\alpha_k = \alpha$ and $\vert S_k \vert = n_k$, satisfies
\begin{equation} \label{eq:prop2}
\begin{split}
&\min_{0 \le i \le {k-1}} \mathbb{E}\left[ \| \nabla F(\bbx_{i})\|_2^2 \right] \le \frac{1}{k\hat{r}(\alpha)} \left( F(\bbx_{0}) - F(\bbx^*) \right)+ \sum_{i = 0}^{k-1}\frac{ \alpha^2 L w}{2kn_i \hat{r}(\alpha) }
\end{split}
\end{equation}
where $\hat{r}(\alpha) = \alpha-\alpha^2 L/2$ and the expectation $\mathbb{E}[\cdot]$ is over the unknown distribution of $\bbxi$ whose samples $\ccalS_k = \{\bbxi_j \}_{j=1}^{n_k}$ are observed at each iteration $k$.
\end{proposition}

\begin{proof}
See Appendix \ref{pr:prop2}.
\end{proof}

Proposition \ref{prop2} characterizes the dependence of convergence rate on problem constants as well as the step-size and the batch-size. Specifically, an optimal step-size can be selected for a minimal value of first term, and the second term is inversely dependent on the batch-size. For the special case with constant $n_k = n$, we have
\begin{align} \label{eq:prop21}
&\min_{0 \le i \le {k-1}} \mathbb{E}\left[ \| \nabla F(\bbx_{i})\|_2^2 \right] \le \underbrace{\frac{1}{\alpha k-\frac{\alpha^2 L k}{2}} \left( F(\bbx_{0}) - F(\bbx^*)\right) }_{:=Q_1}+ \underbrace{ \frac{ \alpha L w}{2 n \left(1-\frac{\alpha L}{2}\right)}.}_{:=Q_2}
\end{align}
Similarly to \eqref{eq:coro1main} for convex problems, the bound in \eqref{eq:prop21} consists of two terms. The convergence rate term $Q_1$ decreases as iteration $k$ increases. The error neighborhood term $Q_2$ indicates the limiting radius of convergence, which stays constant as long as $n$ and $\alpha$ are constant. Both two terms depends on the batch-size $n$ and step-size $\alpha$ jointly. Based on these observations, the TSA scheme can be adjusted appropriately for non-convex problems.

\subsection{Non-convex Two Scale Adaptive Scheme}

From \eqref{eq:prop21}, we observe that $Q_1$ decreases monotonically while $Q_2$ keeps constant as iteration $k$ increases. When $Q_1$ decays below $Q_2$, the convergence accuracy cannot be tightened beyond $Q_2$. In this case, either the decreasing of step-size $\alpha$ or the increasing of batch-size $n$ is required at some point to reduce the radius of convergence. To surmount this issue, the non-convex TSA scheme consists of two scales: the inner time-scale performs SGD with constant step-size $\alpha$ and batch-size $n$, and the outer time-scale tunes parameters to tighten the convergence radius. Similarly, we develop two sub-schemes, the post TSA and the prior TSA, with different inner time-scale stop criterions.

\emph{Initialization}: Let $\alpha = \alpha_0$ and $n = n_0$ be the initial step-size and batch-size. Define $Q_1^0$, $Q_2^0$ as in \eqref{eq:prop21} and assume $F(\bbx_{0})-F(\bbx^*)$ is large such that initially $Q_1^0 \ge Q_2^0$. The multiplicative term $1/(\alpha_0 k - \alpha_0^2 L k/2)$ in $Q_1^0$, whose denominator is quadratic, is minimized at $2L/k$ with step-size $\alpha_0 = 1/L$. For this selection, $Q_2^0$ is
\begin{equation}
\begin{split}
Q_2^0 = \frac{ \frac{1}{L} L w}{2 n_0 \left(1-\frac{\frac{1}{L} L}{2}\right)}= \frac{w}{ n_0}.
\end{split}
\end{equation}
Our intention is to reduce the bound in \eqref{eq:prop21} for fixed $\alpha=1/L$ that minimizes $Q_1$ over all iterations. Then, we adapt the batch-size $n$ based on the following criterion to tighten the error neighborhood $Q_2$ continuously to null.

(1-A) \emph{Inner time-scale (post)}. At $t$-th inner time-scale, let $\alpha_{t} = 1/L$, $n_{t}$ and $K$ be the current step-size, batch-size and the beginning number of iterations. Proposition \ref{prop2} for this selection allows us to write
\begin{equation} \label{nonconvex1}
\begin{split}
\min_{0 \le i \le {k-1}} \mathbb{E}\left[ \| \nabla F(\bbx_{i})\|_2^2 \right]&\le \frac{2L}{k_t+K} \left( F(\bbx_0)-F(\bbx^*) \right) + \sum_{i=0}^{t-1}\frac{K_i w}{\left( k_t+K \right) n_i} + \frac{ w}{n_t}
\end{split}
\end{equation}
where $\{ n_0, n_1,..., n_{t-1} \}$ and $\{ K_0, K_1,..., K_{t-1} \}$ are batch-sizes and durations of previous inner time-scales, and $k_t$ is the passed number of iterations at $t$-th inner time-scale. The inequality $ k_tw/( k_t+K ) n_i \le w/n_t $ is also utilized in \eqref{nonconvex1}. Based on \eqref{nonconvex1}, we define $Q_1^t$ and $Q_2^t$ as
\begin{align}\label{eq:nonerror_conditions20}
&Q_1^t = \frac{1}{k_t+K} \left( 2L\left( F(\bbx_0)-F(\bbx^*) \right) +\sum_{i=0}^{t-1}\frac{K_i w}{ n_i} \right), \\
\label{eq:nonerror_conditions21}&Q_2^t = \frac{w}{n_t}.
\end{align}
Here, $Q_1^t$ keeps decreasing as $k_t$ increases, while $Q_2^t$ stays constant. There then exists $K_t$ such that
\begin{equation} \label{eq:noninnercond2}
\begin{split}
K_t =  \max_{k_t} \left\{ Q_1^t \ge Q_2^t \right\}.
\end{split}
\end{equation}
$K_t$ is the duration of $t$-th inner time-scale, which is the largest iteration before $Q_1^t$ drops below $Q_2^t$. Since $\{ n_0, n_1,..., n_{t-1} \}$ and $\{ K_0, K_1,..., K_{t-1} \}$ are historical information available from previous stages, the stop criterion \eqref{eq:noninnercond2} is ready for the implementation. Thus, the $t$-th inner time-scale runs SGD with step-size $\alpha_t = 1/L$ and batch-size $n_t$ for $K_t$ iterations to reduce $Q_1^t$. Right before $Q_1^t$ drops below $Q_2^t$, we increment the batch-size at $t$-th outer time-scale.

{\linespread{1}
\begin{algorithm}[t] \begin{algorithmic}[1]
\STATE \textbf{Input:} objective functions $\{ f(\bbx,\bbxi) \}$, decision vector $\bbx_0$
\STATE Set step-size $\alpha = 1/L$; sample-size $\vert S_0 \vert=n_0$ and $t=0$
\STATE Update the decision vector $\bbx_{1} = \bbx_0 - \alpha \nabla f_{S_k}(\bbx_0)$
\STATE Compute $Q_1^t =  2L (F(\bbx_0)-F(\bbx^*))$, $Q_2^t = \frac{w}{n_0}$
\STATE Update the decision vector $\bbx_{1} = \bbx_0 - \alpha \nabla f_{S_0}(\bbx_0)$
\FOR [main loop]{$k = 1,2...$}
	  \IF {$\frac{k}{k+1}Q_1^t \le Q_2^t$}
	      \STATE Update batch-size $n_{t+1} = n_t + \beta$ or $n_{t+1} = m n_t$
	      \STATE Set $K_t = k - \sum_{i=0}^{t-1}K_i$, $t=t+1$
	      \STATE Post: Update $Q_1^{t} = Q_1^{t-1} + \frac{K_{t-1} w}{k n_{t-1}}$, $Q_2^{t} = \frac{n_{t-1} Q_2^{t-1}}{n_{t}}$
	      \STATE Prior: Update $Q_2^{t} = \frac{n_{t-1} Q_2^{t-1}}{n_{t}}$
      \ENDIF
      \STATE Update the decision vector $\bbx_{k+1} = \bbx_k - \alpha \nabla f_{S_k}(\bbx_k)$
	  \STATE Compute $Q_1^t = \frac{k}{k+1} Q_1^t$
\ENDFOR
\end{algorithmic}
\caption{Two Scale Adaptive Scheme (Non-convex Case)}\label{alg:learning} \end{algorithm}}

(1-B) \emph{Inner time-scale (prior)}. At $t$-th inner time-scale, let $\alpha_{t} = 1/L$, $n_{t}$ and $K$ be the current step-size, batch-size and the beginning number of iteration. We now define $Q_1^t$ and $Q_2^t$ based on \eqref{nonconvex1} as
\begin{align}\label{eq:nonerror_conditions3}
&Q_1^t = \frac{2L}{k_t+K} \left( F(\bbx_0)-F(\bbx^*) \right),\\
&Q_2^t = \frac{w}{n_t}
\end{align}
where $k_t >0$ is the passed number of iterations at $t$-th inner time-scale. It follows the same rule that $Q_1^t$ decreases with the increasing of $k_t$, while $Q_2^t$ is constant. Then the duration of $t$-th inner time-scale $K_t$ is defined as the same as \eqref{eq:noninnercond2}. We stop the $t$-th inner time-scale when $k_t$ reaches $K_t$ and augment the batch-size as detailed next.

\begin{remark}\normalfont
The post (1-A) and the prior (1-B) sub-schemes are different, where $Q_1^t$ of the post contains accumulated variance errors of previous inner time-scales (the second term in the bound \eqref{nonconvex1}), while $Q_1^t$ of the prior only utilizes the initial sub-optimality error. Therefore, the prior converges faster while the post saves more sample complexity.   
\end{remark}

(2) \emph{Outer time-scale}. At $t$-th outer time-scale, we reduce the error neighborhood $Q_2^t$ with parameter tuning. This may be done by decreasing the step-size $\alpha_t$ or increasing the batch-size $n_t$, where the former slows the decrease of $Q_1$ and the latter increases the sample complexity. We propose increasing the batch-size $n_t$ to $n_{t+1}$ additively \eqref{eq:add} or multiplicatively \eqref{eq:multi}. Once the batch-size has been increased, we proceed to the $(t+1)$-th inner time-scale.

Here we also require the duration $K_t$ larger than zero by appropriate selections of $\beta$ and $m$. As previously mentioned, one can approximate the initial error $F(\bbx_0)-F(\bbx^*)$ with $\max_{\bbx, \bby} \vert F(\bbx) - F(\bby) \vert$ in practical experiments. We show the non-convex TSA scheme in Algorithm 2 with $t$ the number of inner/outer time-scales and $k$ the number of iterations.

\section{Performance Analysis}\label{sec:convergence}

In this section, we analyze the performance of the TSA scheme in terms of its convergence, rate and sample complexity for both convex and non-convex problems. In particular, we establish that it inherits the limiting properties of SGD with a fast rate, while reducing the number of training samples required to reach the $\epsilon$-suboptimality.

\subsection{Convex Problems}

Consider the convex case first. We show that the sequence of objective function values $F(\bbx_k)$ generated by the TSA scheme approaches the optimal value $F(\bbx^*)$ with the following theorem, which guarantees the exact convergence of TSA. 

\begin{theorem} \label{eq:thm1}
Consider the post and the prior TSA schemes for convex problems. If the objective functions satisfy Assumptions 1-3, both sequences of $ F(\bbx_k)$ and $\bbx_k$ converge to the optimal $F(\bbx^*)$ and $\bbx^*$ almost surely, i.e.,
 \begin{equation}
\begin{split}
&\lim_{k \to \infty} \mathbb{E}\left[ F(\bbx_k) - F(\bbx^*) \right] = 0,\\
&\lim_{k \to \infty} \mathbb{E}\left[ \| \bbx_k - \bbx^* \|_2 \right] = 0.
\end{split}
\end{equation}
\end{theorem}

\begin{proof}
See Appendix \ref{pr:thm1}.
\end{proof}

Theorem \ref{eq:thm1} shows that TSA inherits the asymptotic convergent behavior of SGD with diminishing step-size. However, this result is somewhat surprising since TSA attains exact convergence with a constant step-size, whereas SGD converges to an error neighborhood under this setting. The preserved constant step-size then helps TSA maintain a fast rate.

We then characterize the convergence rate of TSA, as stated in the following.

 \begin{theorem} \label{prop3}
Consider the TSA scheme for convex problems. If the objective functions satisfy Assumptions 1-3, the post and the prior TSA schemes converge approximately with rates of $\ccalO(2^t(1-\ell/L)^k)$ and $\ccalO(t(1-\ell/L)^k)$, where $k$ is the number of iterations and $t$ is the number of inner/outer time-scales. If particularizing the multiplicative rule \eqref{eq:multi} for augmenting the batch-size, the post and the prior converge approximately with rates of $\ccalO\left(m^{-k/\log_{1-\ell/L}(1/2m)}\right)$ and $\ccalO\left((1-\ell/L)^kk/\log_{1-\ell/L}(1/m)\right)$.
\end{theorem}

\begin{proof}
See Appendix \ref{pr:prop3}.
\end{proof}

Theorem \ref{prop3} establishes that both the post and the prior TSA schemes have favorable rates under the premise of exact convergence. In particular, the prior obtains a faster rate but requires more samples per iteration. While the post needs fewer samples with the drawback of converging more slowly. Two sub-schemes imply two kinds of balances with different preferences between rates and variance. This result is a precursor to characterizing the sample complexity of TSA, which we do next.

%

We compare the sample complexity of TSA with SGD for an $\epsilon$-suboptimal solution. In order to simplify expressions and make a comparison possible, we hypothesize that both TSA and SGD make use of the optimal step-size $\alpha = 1/L$. The TSA uses the multiplicative rule \eqref{eq:multi} for augmenting the batch-size. Under these conditions, the sample complexity of TSA compared with SGD may be derived.

\begin{theorem} \label{thm2}
Consider the TSA scheme starting with the initial batch-size $n_0=1$, and the SGD with constant step-size $\alpha = 1/L$ and batch-size $n$. Under Assumptions \ref{asm1}-\ref{as3}, define the initial error $D:=F(\bbx_0)-F(\bbx^*)$. To achieve an $\epsilon$-suboptimal solution, the ratio $\gamma$ between the number of training samples required for TSA and SGD is
\begin{equation}\label{eq:thm2res}
\gamma \le 
\begin{cases}
\frac{ \frac{m}{m-1} \left\lceil \log_{1-\frac{\ell}{L}} \frac{1}{2m} \right\rceil+1}{  \left\lceil \log_{1-\frac{\ell}{L}}\frac{\epsilon}{2D} \right\rceil} + \mathcal{O}(\epsilon),  \!&\! \mbox{if the post,} \\
\! \frac{ \frac{m}{m\!-\!1} \!\left\lceil \!\log_{1\!-\!\frac{\ell}{L}} \frac{1}{m} \!\right\rceil\!+\! \left\lceil\!\log_{1 \!-\!\frac{\ell}{L}}\! \frac{1}{\log_{m}\!\left\lceil\! \frac{w}{\ell\epsilon}\! \right\rceil\!+\!1}\! \right\rceil\!+\!1\!}{  \left\lceil \log_{1-\frac{\ell}{L}}\frac{\epsilon}{2D} \right\rceil} + \mathcal{O}(\epsilon)\!, \!&\! \mbox{if the prior}
\end{cases}
\end{equation}
where $\lceil \cdot \rceil$ is the ceil function.
\end{theorem}

\begin{proof}
See Appendix \ref{pr:thm2}.
\end{proof}

Observe from Theorem \ref{thm2}, the ratio of sample complexity of TSA to SGD $\gamma$, is approximately proportional to $\ccalO(-1/\log\epsilon) + \ccalO(\epsilon)$, meaning that for a large $\epsilon$, the complexity reduction is not much. However, for more accurate solutions, i.e., $\epsilon$ close to null, the logarithmic term dominates and a significant reduction in the number of required samples may be attained. Due to the complicated dependence on problem constants, we present a corollary for the post TSA with $m=2$ which simplifies the expressions and provides more intuitive descriptions.
\begin{corollary} \label{coro2}
With the same settings of Theorem \ref{thm2}, to achieve an $\epsilon$-suboptimal solution, the ratio $\gamma$ between the number of training samples required for the post TSA and SGD is
\begin{equation} \label{eq:thm2res1}
\gamma \le \frac{\left\lceil \log_{1-\frac{\ell}{L}} \frac{(L-\ell)^2}{16L^2} \right\rceil}{ \left\lceil \log_{1-\frac{\ell}{L}}\frac{\epsilon}{2D} \right\rceil} + \mathcal{O}(\epsilon)
\end{equation}
\end{corollary}
\begin{proof}
See Appendix \ref{pr:coro2}.
\end{proof}

More concrete takeaways can be discerned from Corollary \ref{coro2}. For a small $\epsilon$, the second term $\mathcal{O}(\epsilon)$ of \eqref{eq:thm2res1} is negligible. Moreover, the simplified first term indicates that provided the rate $1-\ell/L<1$, the ratio $\gamma$ between required training samples of TSA and SGD is less than $1$ as long as $\epsilon/(2D) \le (L-\ell)^2/(16L^2)$, i.e., $\gamma \le 1$ whenever $\epsilon \le D(1-\ell/L)^2/8$, which is almost always true unless the initial point is very close to the optimizer. Therefore, in practice, the sample complexity of TSA is largely reduced compared with SGD, while both of them achieve the same suboptimal solution. The magnitude of the reduction depends on problem constants, but is proportional to the sum of minus inverse of logarithmic factor of $\epsilon$ and $\epsilon$, which may be substantial.

\begin{remark} \normalfont
The batch-size $n$ of SGD in Theorem \ref{thm2} needs to be preset based on the required suboptimality $\epsilon$ at the outset, i.e., $n$ is determined by $\epsilon$. Otherwise, if $n$ is set smaller than that determined by $\epsilon$, SGD may never achieve an $\epsilon$-suboptimal solution due to the large error neighborhood. While if $n$ is set larger than that determined by $\epsilon$, it will waste more samples in the training. Therefore, it indicates another disadvantage of SGD. TSA overcomes this issue with no need to determine any parameter at the beginning. 
\end{remark}

\subsection{Non-convex Problems}

The TSA scheme for non-convex problems also exhibits comparable asymptotic properties. Specifically, Theorem \ref{thm3} shows that the sequence of decision vectors $\bbx_k$ generated by the non-convex TSA scheme converges to a stationary solution $\bbx^*$ in expectation.

\begin{theorem} \label{thm3}
Consider the post and the prior TSA schemes for non-convex problems. If the objective functions satisfy Assumptions \ref{asm1}-\ref{as2}, the sequence of $\bbx_k$ converges to a stationary solution $\bbx^*$ in expectation, i.e.,
 \begin{equation}
\begin{split}
&\lim_{k \to \infty} \mathbb{E}\left[ \| \nabla F(\bbx_{k})\|_2^2 \right] = 0.
\end{split}
\end{equation}
\end{theorem}

\begin{proof}
See Appendix \ref{pr:thm3}
\end{proof}

Theorem \ref{thm3} shows that the TSA scheme attains the exact convergence as the SGD with attenuating step-size for non-convex problems. We follow to establish the convergence rate of non-convex TSA. Note that since $Q_1^t$ of the post TSA in \eqref{eq:nonerror_conditions20} contains historical information about durations of previous inner time-scales, it is challenging to theoretically analyze the rate of the post TSA. Thus, we focus on the prior TSA here.

\begin{theorem} \label{prop4}
Consider the prior TSA scheme for non-convex problems. If the objective functions satisfy Assumptions \ref{asm1}-\ref{as2}, the prior TSA converges approximately with a rate of $\ccalO(\log k /k)$.
\end{theorem}

\begin{proof}
See Appendix \ref{pr:prop4}.
\end{proof}

\begin{figure*}%
\centering
\begin{subfigure}{0.33\columnwidth}
\includegraphics[width=1.1\linewidth, height = 0.83\linewidth]{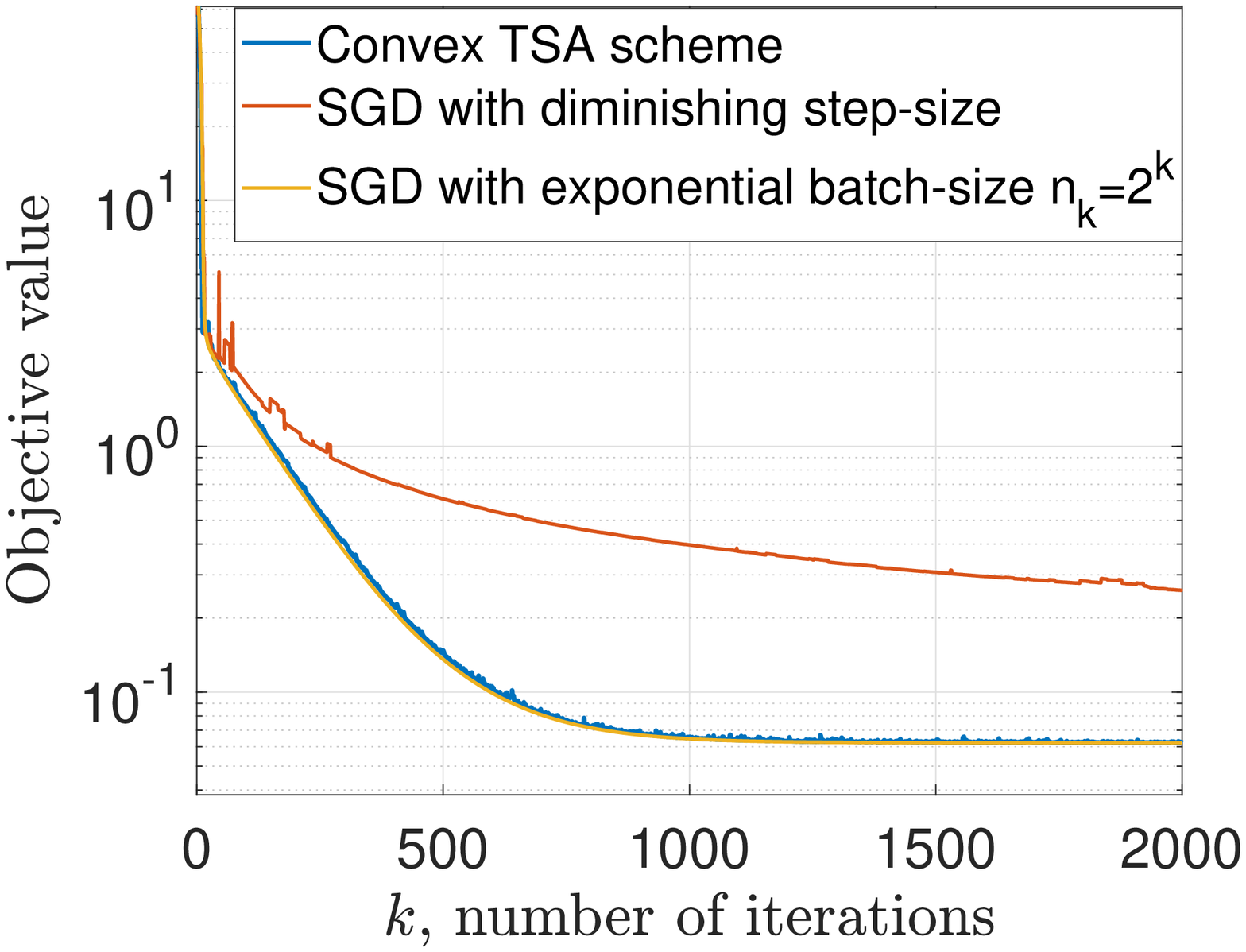}%
\caption{}%
\label{subfiga_vary_n}%
\end{subfigure}\hfill\hfill%
\begin{subfigure}{0.33\columnwidth}
\includegraphics[width=1.1\linewidth,height = 0.83\linewidth]{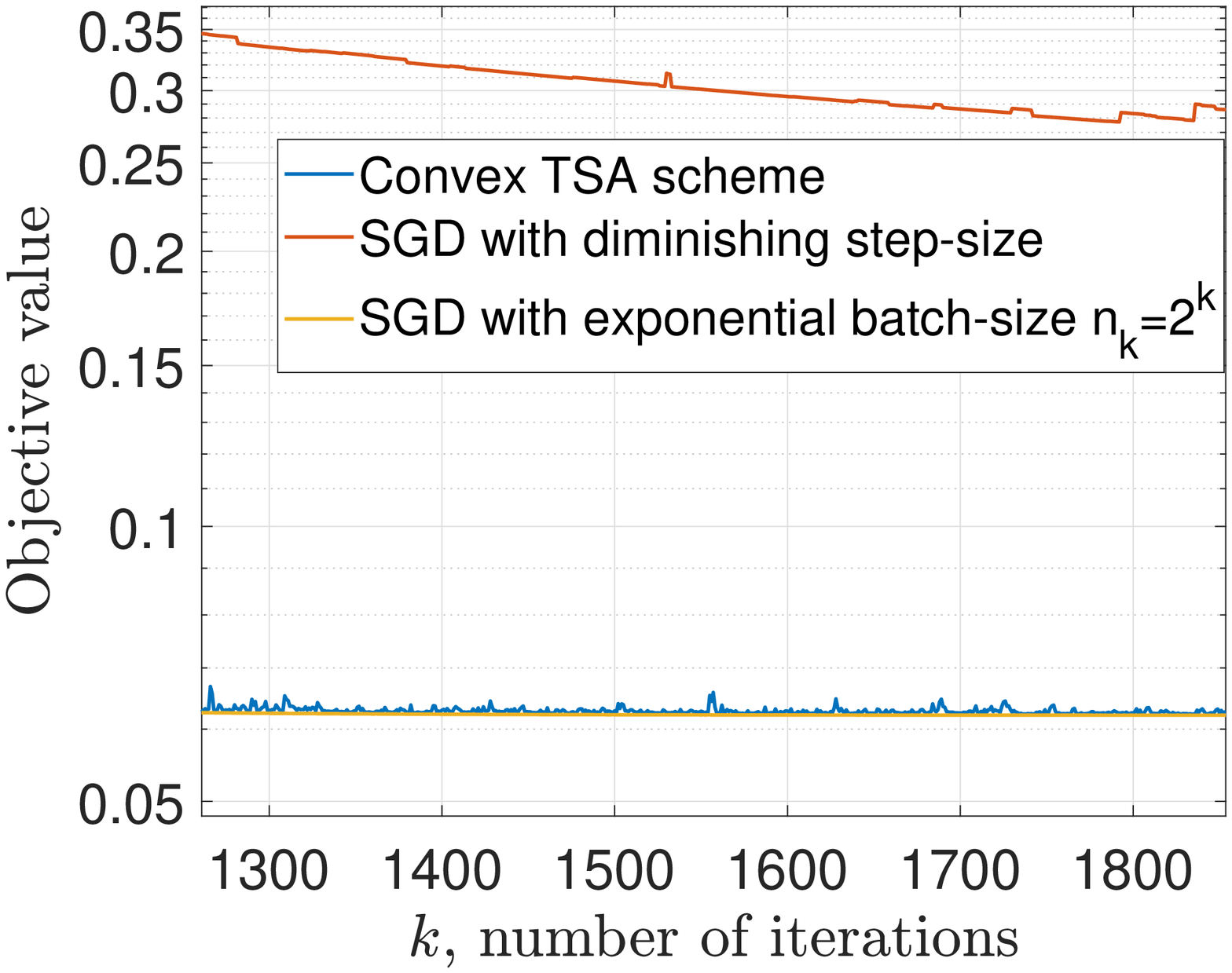}%
\caption{}%
\label{subfigb_vary_n}%
\end{subfigure}\hfill\hfill%
\begin{subfigure}{0.33\columnwidth}
\includegraphics[width=1.05\linewidth,height = 0.83\linewidth]{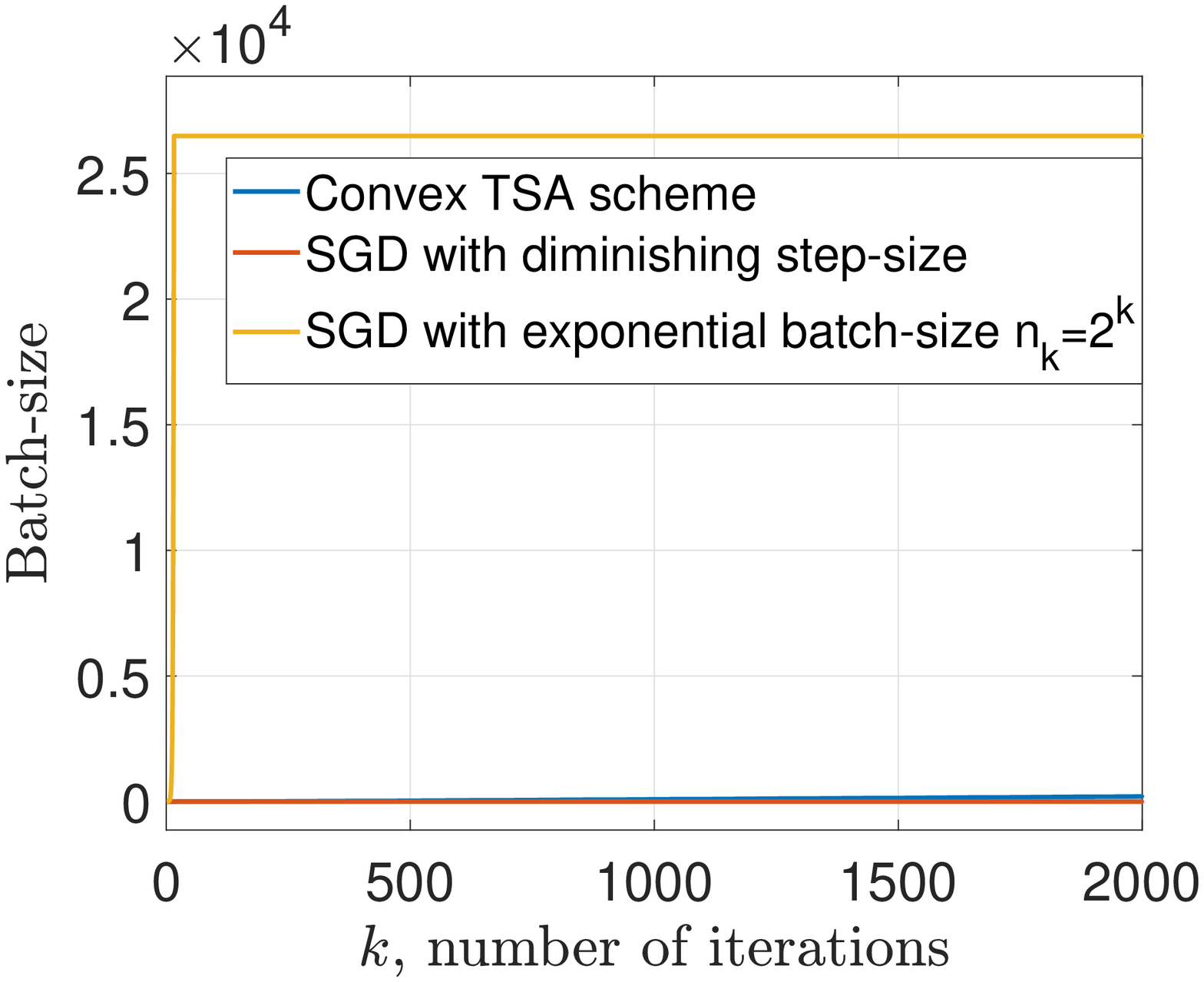}%
\caption{}%
\label{subfigc_vary_n}%
\end{subfigure}%
\caption{Linear modle on MNIST logistic regression with three exact convergent algorithms: convex TSA, SGD with diminishing step-size $\alpha_k = \min(1/L, 100/(Lk))$ and SGD with exponential batch-size $n_k=2^k$: (a) Objective value as a function of training iterations; (b) The zoomed version of (a) showing detailed performance differences; (c) Batch-size required per iteration as a function of training iterations. }\label{fig:vary_n}
\end{figure*}

Proposition \ref{prop4} establishes that the prior TSA, with either the additive or the multiplicative rule for augmenting batch-size, enjoys a much faster rate of $\ccalO(\log k /k)$ compared with $\ccalO(1/\sqrt{k})$ of SGD with diminishing step-size for non-convex problems, while still maintaining the exact convergence. In terms of the post TSA, though with no exact characterization, it converges slightly slowly than the prior but requires less samples per iteration. Together, the post and the prior give two strategies to evolve parameters for the balance between rates and variance.

The exact sample complexity of non-convex TSA is challenging to characterize because $Q_1^t$ of the post TSA in \eqref{eq:nonerror_conditions20} contains historical information and there is no closed-form solution of $\log k /k = \epsilon$ for the prior TSA. However, with observations from Theorem \ref{prop4}, we may refer that the prior TSA converges close to the theoretical optimal rate $\ccalO(1/k)$ of true gradient descent algorithm but will use fewer samples per iteration. Thus, the sample complexity shall be reduced under the exact convergence with such a fast rate; something we will verify in the numerical experiments in Section \ref{sec:sims}.


Overall, the TSA scheme provides a strategy to evolve SGD parameters for both convex and non-convex problems. It selects the step-size and batch-size to preserve a fast convergence rate while repeatedly reducing the stochastic approximation variance during the training process. Under the proposed criterion, the batch-size increases only when necessary, allowing provable sample complexity reduction relative to classical SGD schemes in convex problems. For non-convex problems, it achieves a faster convergence rate than SGD with attenuating step-size under the premise of exact convergence, and saves the sample complexity as much as possible. Together, it well balances the rate and variance in SGD and exhibits an improved performance theoretically. We summarize main results of TSA for convex and non-convex problems in Table \ref{Tab03} and investigate the experimental implications of these results in the subsequent section.

\begin{remark}\normalfont
Two sub-schemes, the post TSA and the prior TSA, imply different emphasis within the balance, where the former acquires a faster rate while the latter achieves lower sample complexity. The choice of which sub-scheme depends on specific problems. For problems with small variance, i.e, target sub-optimality can be obtained without large samples, the post is preferred since it increases the batch-size slowly to save more unnecessarily wasted samples. For problems with large variance, one may then use the prior that increases the batch-size faster for a faster rate.
\end{remark}

\begin{figure*}%
\centering
\begin{subfigure}{0.33\columnwidth}
\includegraphics[width=1.1\linewidth, height = 0.83\linewidth]{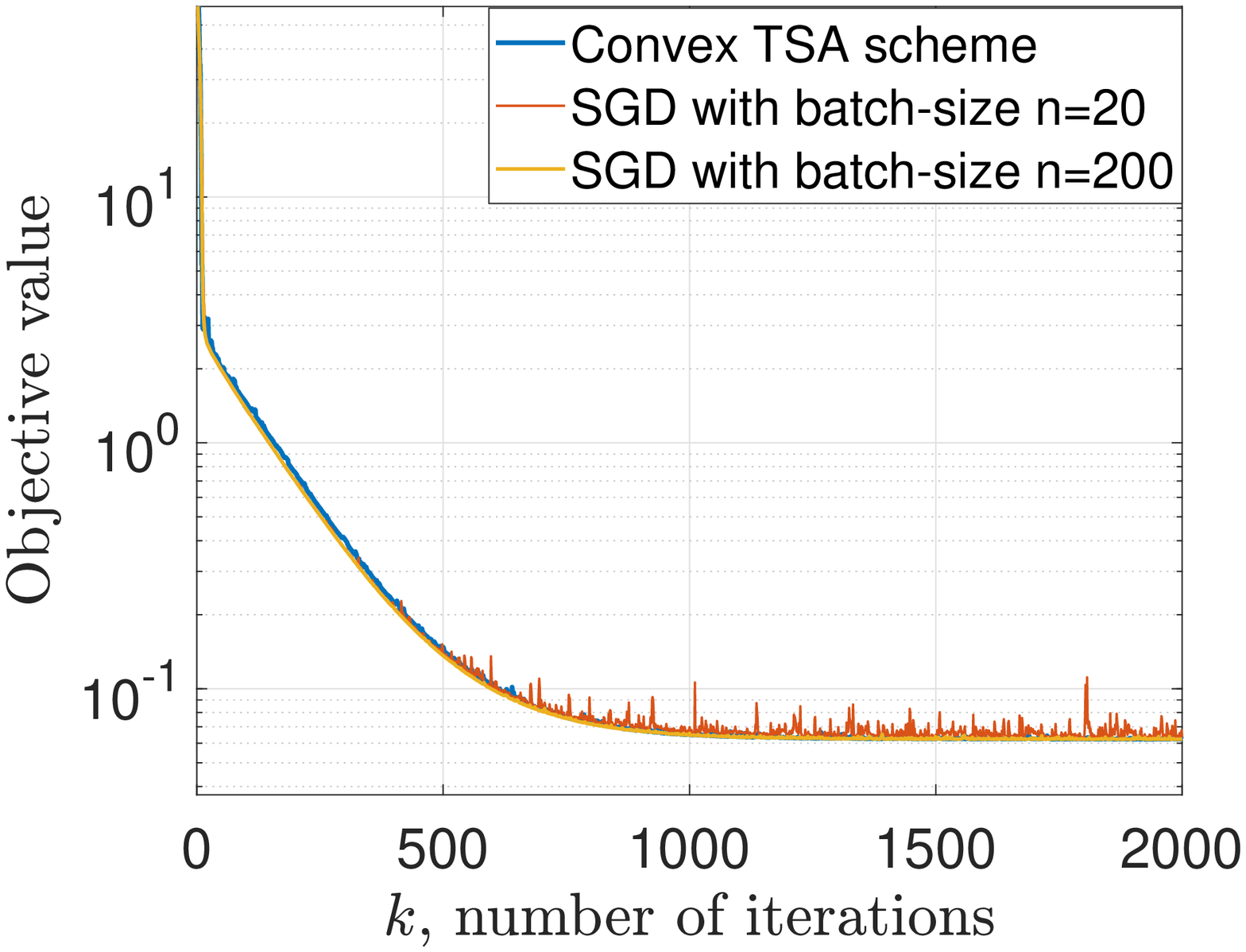}%
\caption{}%
\label{subfiga_convex}%
\end{subfigure}\hfill\hfill%
\begin{subfigure}{0.33\columnwidth}
\includegraphics[width=1.1\linewidth,height = 0.83\linewidth]{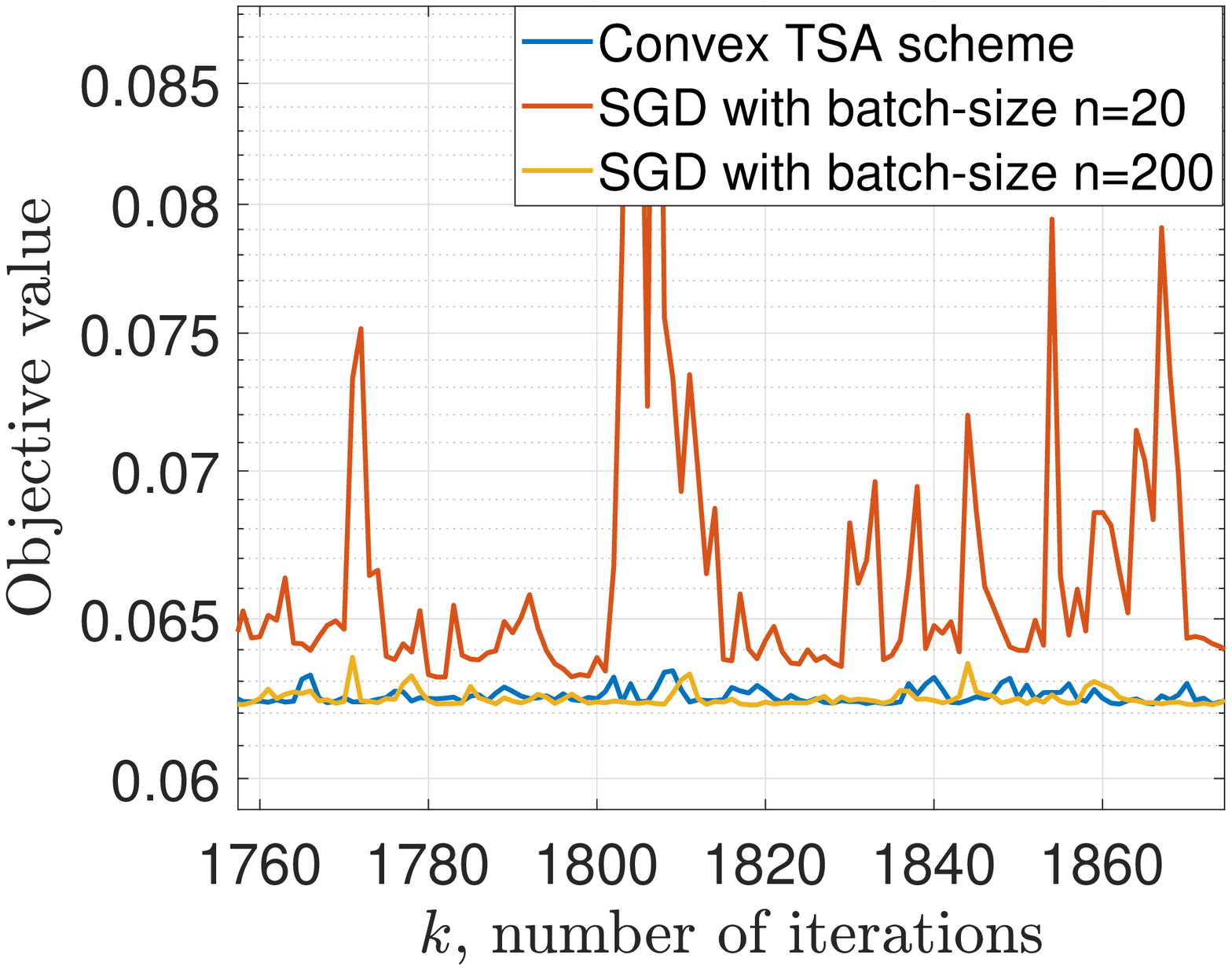}%
\caption{}%
\label{subfigb_convex}%
\end{subfigure}\hfill\hfill%
\begin{subfigure}{0.33\columnwidth}
\includegraphics[width=1.05\linewidth,height = 0.83\linewidth]{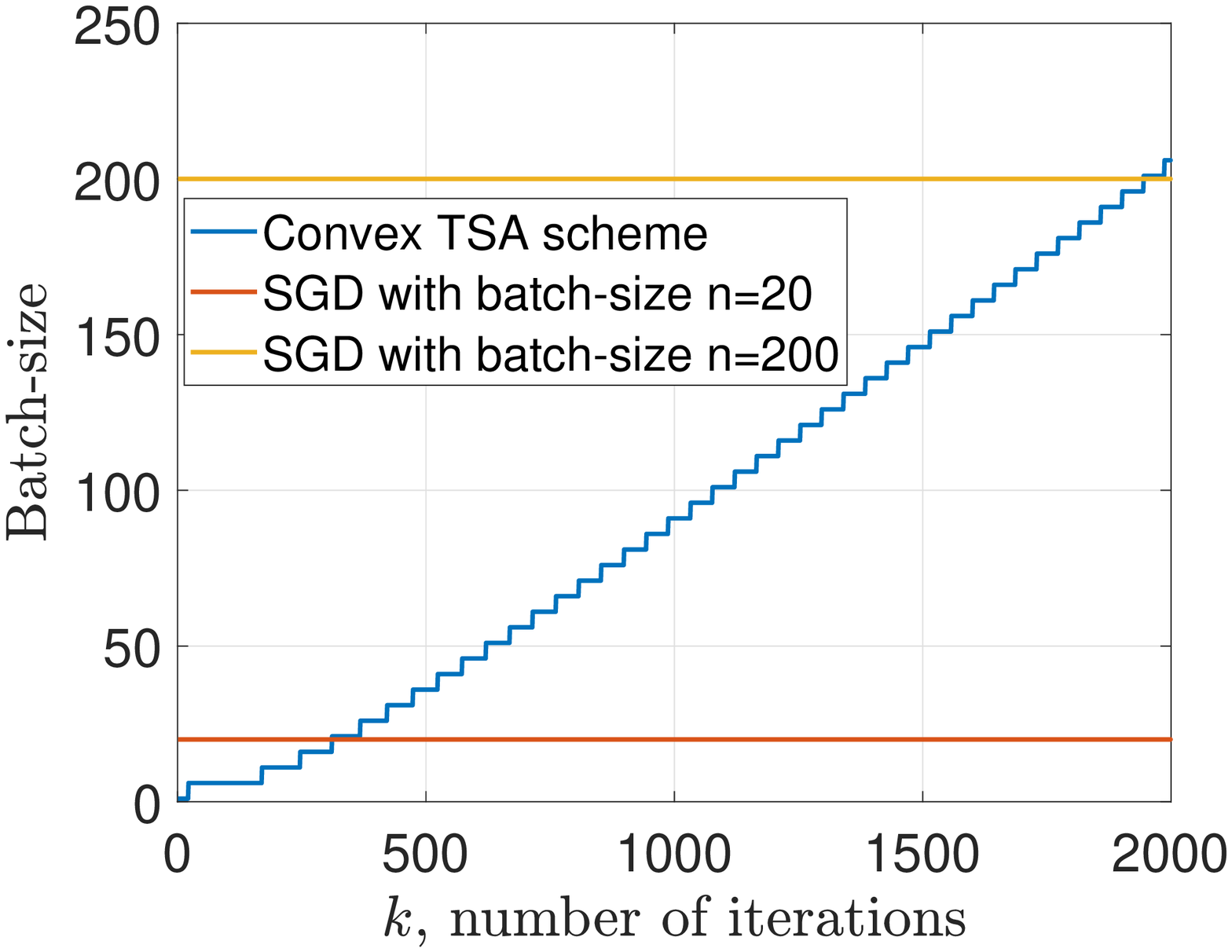}%
\caption{}%
\label{subfigc_convex}%
\end{subfigure}%
\caption{Linear modle on MNIST logistic regression with convex TSA, and two approximate convergent algorithms: SGD with constant batch-size $n=200$ and SGD with constant batch-size $n=20$: (a) Objective value as a function of training iterations; (b) The zoomed version of (a) showing detailed performance differences; (c) Batch-size required per iteration as a function of training iterations. }\label{fig:convex}
\end{figure*}

\begin{table}[t] 
\begin{center}  
\caption{Samples required for loss $0.063$ for SGD with $n=20$, convex TSA, SGD with $n = 200$, and SGD with $n_k=2^k$. Relative sample efficiency computed as the ratio of samples required with respect to SGD with $n_k=2^k$.}  
\label{table2}
\begin{tabular}{|l|l|l|l| p{2cm}|}  
\hline  
Target loss: $0.063$ & Required samples  & Relative sample efficiency \\ \hline  
SGD with $n=20$ & $\infty$ & 0 \\ \hline  
TSA &  60933 & 530 \\  \hline
SGD with $n=200$ & 234800 & 138 \\ \hline
SGD with $n_k = 2^k$ & 32298805 & 1 \\
\hline  
\end{tabular}  
\end{center}  \vspace{-4mm}
\end{table}


\section{Numerical Experiments} \label{sec:sims}


We numerically evaluate the TSA scheme compared with standard SGD schemes. Without particular description, the default step-size of SGD is the same optimal step-size as TSA for clear comparison.

The visual classification problem of hand-written digits is considered for both convex and nonconvex cases on the MNIST data \cite{lecun-mnisthandwrittendigit-2010}. Given the training data $\mathcal{T} = \{ (\bbz_n, y_n) \}_{n=1}^N$, let $\bbz \in \mathbb{R}^p$ be the feature vector of digit image and $y \in \{ 0,1,\dots,c \}$ its associated label denoting which number is written, and $c=9$ denotes the number of classes (minus 1). Denote by $\bbx \in \mathbb{R}^p$ the parameters of a classifier $h(\bbz)$ which models the relationship between features $\bbz$ and label $y$. In particular, for the convex case, we consider a linear model $h(\bbz)=\bbx^T\bbz$ whose parameters $\bbx$ define a logistic regressor. The limitations of linearity restrict our focus to binary classificiation for this class. By contrast, for the non-convex case, we model the classifier as a two layered convolutional neural network (CNN) with the ReLu nonlinearity and the MaxPooling followed by a fully connected layer, and consider the full multi-class problem. The MNIST data is such that the dimension of features is $p=784$ and total sample number is $N=26491$.

\subsection{Linear Model}

As previously mentioned, we restrict focus to classifying digits 0 and 8 in the linear model. The expected objective function $F(\bbx)$ with $\bbxi = (\bbz, y)$ in \eqref{eq:main_prob} is defined as the $\lambda$-regularized negative log-likelihood
\begin{equation}
\begin{split}
F(\bbx) = \frac{\lambda}{2} \| \bbx \|^2 + \frac{1}{N} \sum_{n=1}^N \log \left(1+\exp (-y_n \bbx^\top \bbz_n)\right)
\end{split}
\end{equation}
where $(\lambda/2) \| \bbx \|^2$ is the regularization term. Note that it is actually an ERM problem, which is an instantiation of  \eqref{eq:main_prob}. We run the TSA scheme on this objective with the understanding that it applies more broadly to the population problem. In this case, the variance of stochastic approximation is relatively small such that we do not need a large batch-size to approximate the true gradient. Hence, the post TSA with additive rule \eqref{eq:add} is chosen with $n_0=1$ and $\beta = 5$. 
 
We first run three exact convergent algorithms: TSA, SGD with diminishing step-size $\alpha_k = \min(1/L, 100/(Lk))$ and SGD with exponential batch-size $n_k = 2^k$. Fig. 1 plots the objective value and the batch-size as iteration $k$ increases. SGD with exponential batch-size achieves the best performance (Fig. 1b), but its batch-size explodes quickly to $26491$ indicating huge sample complexity (Fig. 1c). SGD with diminishing step-size requires one sample per iteration such that has the least sample complexity, however, it converges too slowly. Considering TSA, on the one hand, it performs comparably to SGD with exponential batch-size, and its batch-size only grows from 1 to 206, improving the computational cost substantially. Relative to SGD with diminishing step-size, it achieves improved convergence but only requires a small number of additional samples. Overall, though three algorithms all converge exactly, TSA attains comparable convergence accuracy to SGD with exponential batch-size and comparable sample complexity to SGD with diminishing step-size.
 
Fig. \ref{fig:convex} shows performances of TSA and two approximate convergent algorithms. Specifically, TSA and SGD with constant batch-size $n\!=\!200$ exhibit comparable performances, among which the latter is just slightly better. SGD with $n=20$ varies in a large error neighborhood, and performs worse than another two. Fig. \ref{subfigc_convex} depicts the corresponding batch-size versus iteration. Observe that TSA saves more than half of samples compared with SGD with $n=200$, but achieves similar performance. In terms of SGD with $n=20$, though it requires least samples, it performs too badly to consider. It should also be noted that SGDs with $n=200$ and $n=20$ cannot reach the optimal solution but will be trapped in an error neighborhood eventually.

\begin{figure*}%
\centering
\begin{subfigure}{0.5\columnwidth}
\includegraphics[width=0.9\linewidth, height = 0.65\linewidth]{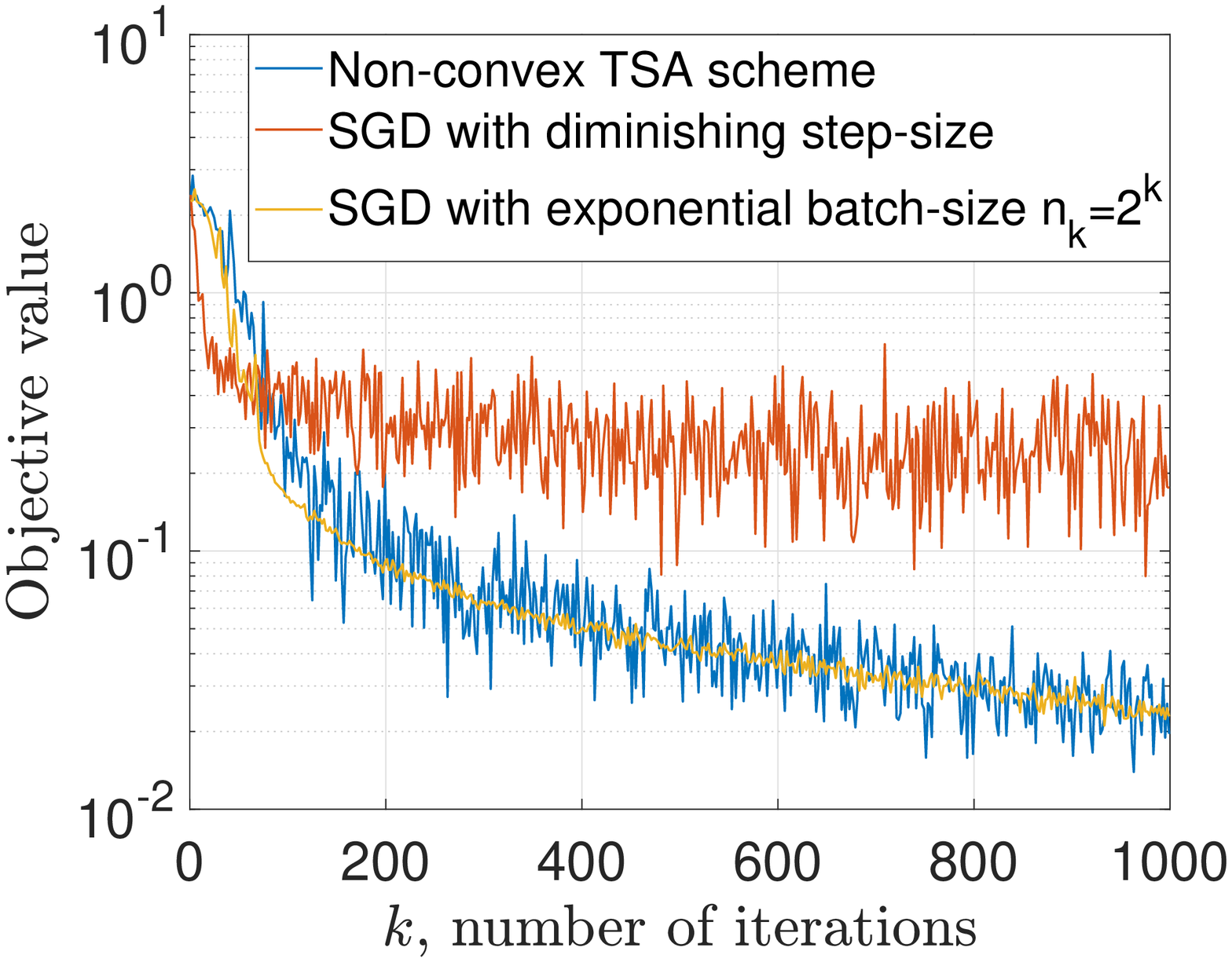}%
\caption{}%
\label{subfiga_nonconvex}%
\end{subfigure}\hfill\hfill%
\begin{subfigure}{0.5\columnwidth}
\includegraphics[width=0.9\linewidth,height = 0.65\linewidth]{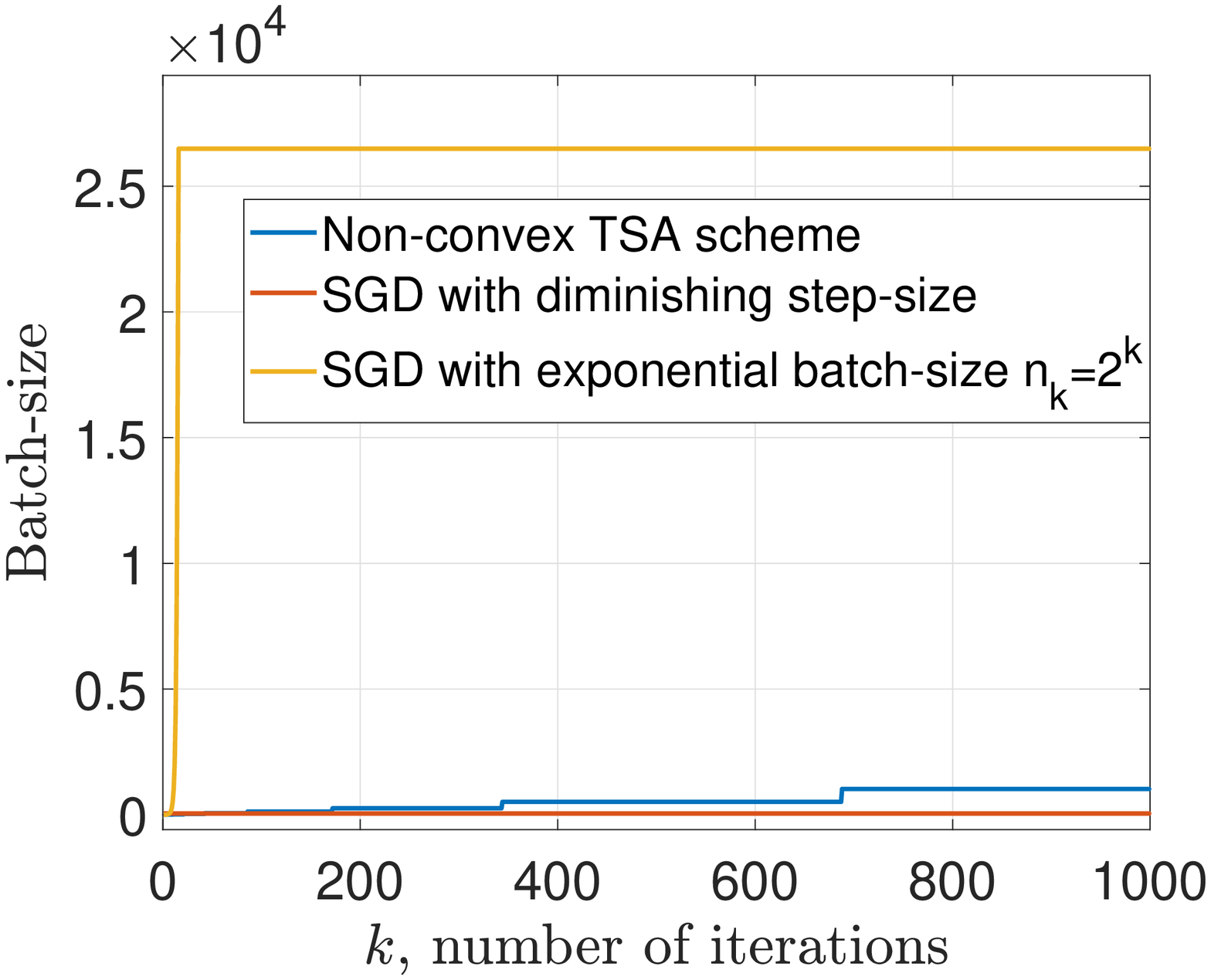}%
\caption{}%
\label{subfigb_nonconvex}%
\end{subfigure}%
\caption{Convolutional neural network on MNIST classification problem with three exact convergent algorithms: non-convex TSA, SGD with diminishing step-size $\alpha_k = \min(1/L, 100/(Lk))$ and SGD with exponential batch-size $n_k=2^k$: (a) Objective value as a function of training iterations; (b) Batch-size required per iteration as a function of training iterations. }\label{fig:nonconvex}
\end{figure*}
  
%
%

To further substantiate these trends, we compare the number of samples required to reduce the loss to a target sub-optimality for four algorithms: TSA, SGDs with $n=20$ and $n=200$, and SGD with exponential batch-size $n_k = 2^k$. Let the target loss be $0.063$. Table \ref{table2} summarizes the required number of samples and the relative sample efficiency (based on SGD with $n_k = 2^k$) for four algorithms. We can see the sample complexity of TSA is far less than SGD with $n=200$ and $n_k = 2^k$, but performs almost as well as them (Fig. 1b and Fig. 2b). SGD with $n=20$ never obtains such a loss due to its large variance error from stochastic approximation. Thus, its number of samples is infinity. Overall, TSA exhibits the best relative sample efficiency among four algorithms. 


\subsection{Convolutional Neural Network}

To analyze how TSA works for non-convex problems, we consider a two-layered convolutional neural network $\Phi(\bbz)$ for multi-class classification of all numbers from $\{0,\dots,9\}$ in MNIST dataset \cite{ciresan2011flexible}. In particular, the first layer contains $25$ filters and the second layer contains $50$ filters, where each filter is with kernel size $3$. The ReLu and the MaxPooling are utilized as the activation and pooling functions, respectively. A fully connected layer follows in the end to match the output dimension. The expected objective function $F(\bbx)$ with $\bbxi = (\bbz, y)$ in \eqref{eq:main_prob} is the cross entropy loss 
\begin{equation}
\begin{split}
F(\bbx) = - \frac{1}{N} \sum_{n=1}^N z_n \log \left(\text{softmax}(\Phi(\bbx, \bby_n))\right)
\end{split}
\end{equation}
where $\text{softmax}(\cdot)$ is the softmax function. As the classification problem becomes more complicated, i.e., classifying $10$ numbers rather than $2$ numbers, and the CNN architecture are more complex, the variance of stochastic approximation in non-convex case is large such that we require large batch-size for the target suboptimality. Thus, we select the prior TSA with multiplicative rule \eqref{eq:multi} with $n_0=1$ and $m = 2$.

\begin{figure*}%
\centering
\begin{subfigure}{0.5\columnwidth}
\includegraphics[width=0.9\linewidth, height = 0.65\linewidth]{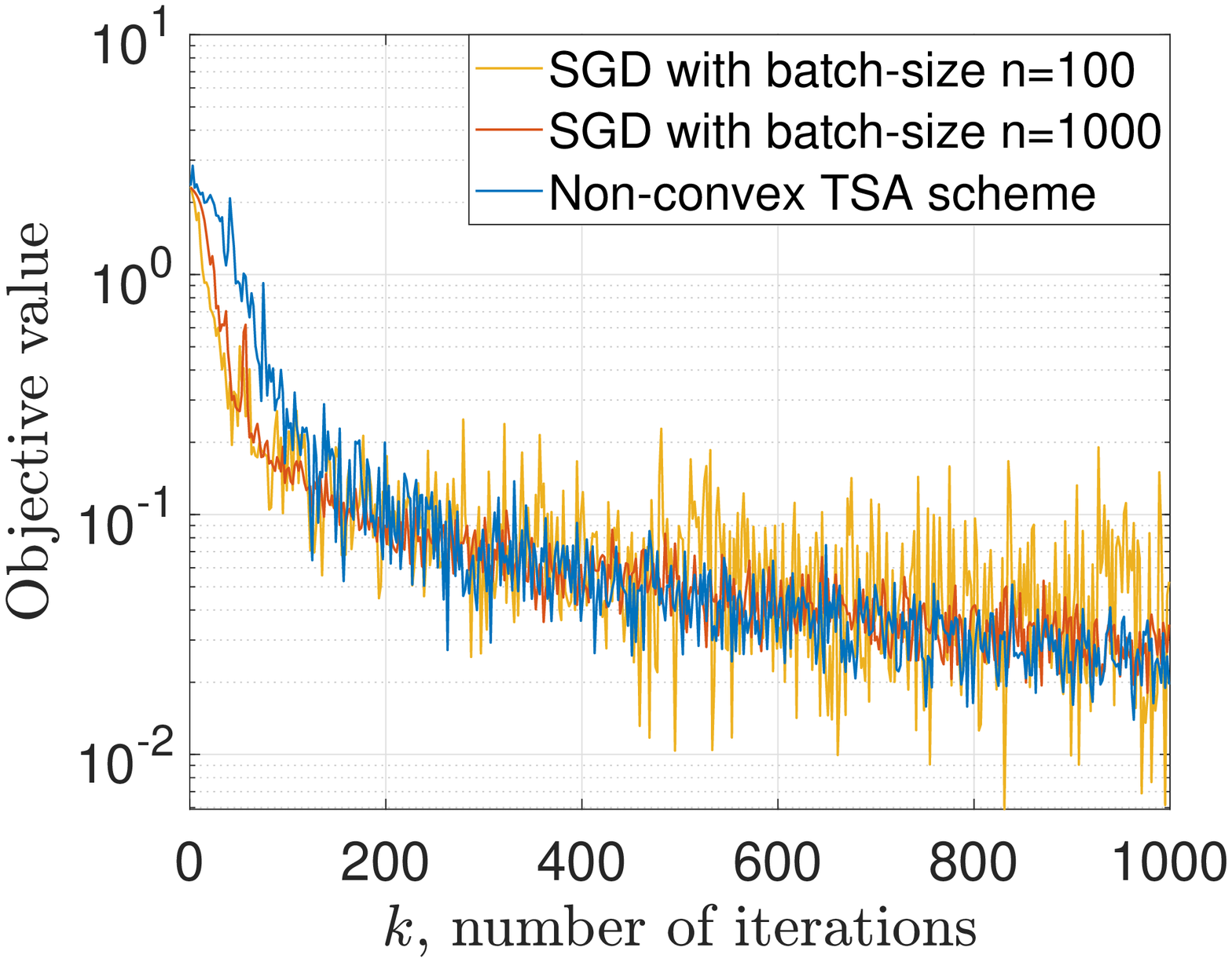}%
\caption{}%
\label{subfiga_nonconvex2}%
\end{subfigure}\hfill\hfill%
\begin{subfigure}{0.5\columnwidth}
\includegraphics[width=0.9\linewidth,height = 0.65\linewidth]{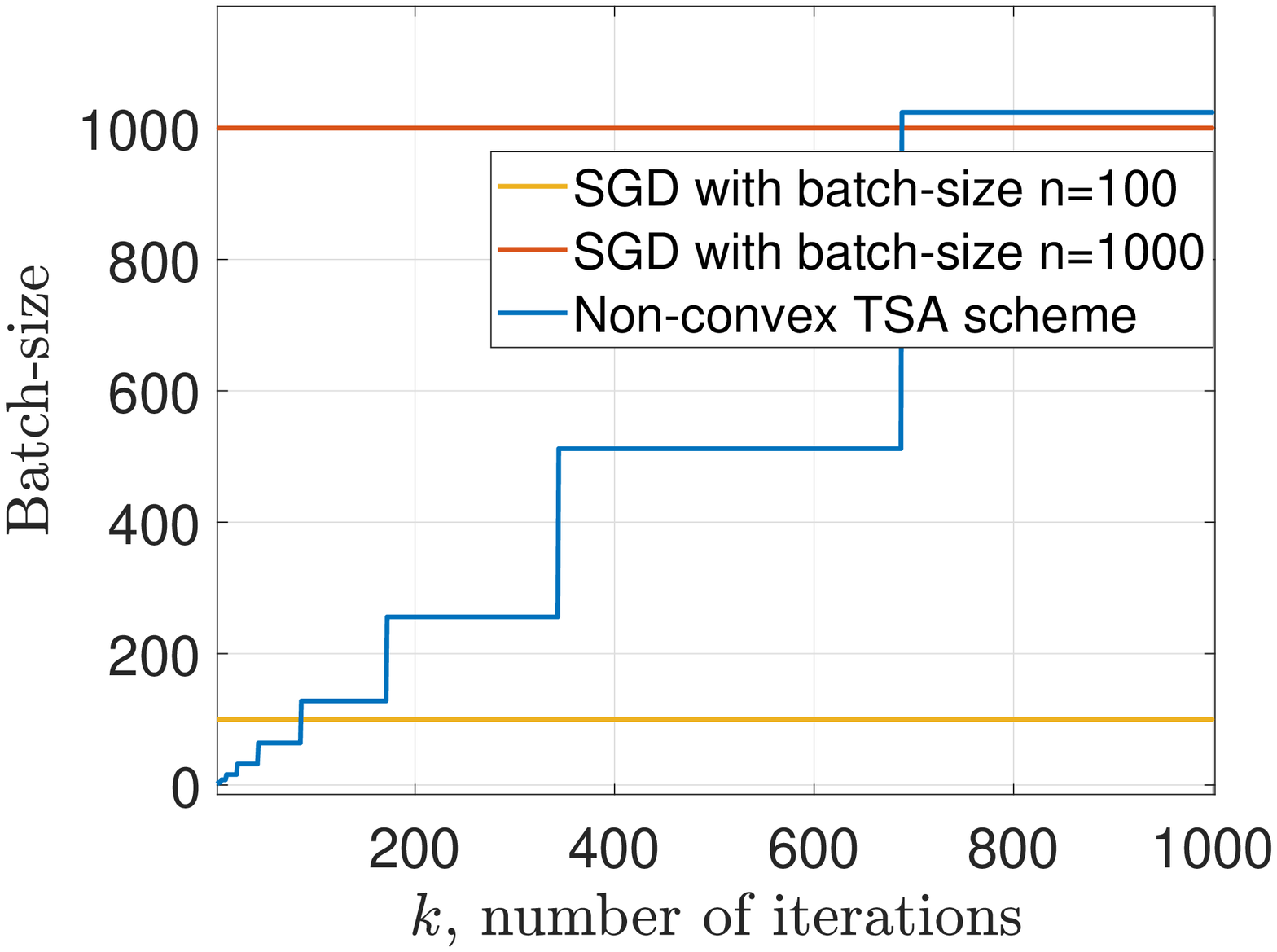}%
\caption{}%
\label{subfigb_nonconvex2}%
\end{subfigure}%
\caption{Convolutional neural network on MNIST classification problem with non-convex TSA, and two approximate convergent algorithms: SGD with constant batch-size $n=100$ and SGD with constant batch-size $n=1000$: (a) Objective value as a function of training iterations; (b) Batch-size required per iteration as a function of training iterations.}\label{fig:nonconvex2}\vspace{-2mm}
\end{figure*}

We first compare three exact convergent algorithms in Fig. 3. We see that TSA and SGD with exponential batch-size $n_k=2^k$ show comparable convergence rates. Though SGD with exponential batch-size performs better with smaller error neighborhood, it wastes too much samples compared with TSA as observed in Fig. 3b. SGD with diminishing step-size $\alpha_k\!=\! \min(1/L, 100/(Lk))$ only requires $50$ samples per iteration with the least sample complexity, while converges slowly due to the reduction of step-size. TSA reaches a good balance between the rate and the sample complexity, i.e., it decreases the objective at a comparable rate to SGD with exponential batch-size and increases the batch-size only when necessary that maintains the comparable sample complexity to SGD with diminishing step-size.

In Fig. 4, we depict performances of TSA and two approximate convergent algorithms. On the one hand, Fig. \ref{subfiga_nonconvex2} shows that TSA performs comparably to SGD with constant batch-size $n=1000$ with similar rates and error neighborhoods, while SGD with $n=100$ sinks into a large error neighborhood and has stopped getting progresses since early iterations. On the other hand, TSA saves almost a half sample complexity compared with SGD with $n=1000$ as observed in Fig. \ref{subfigb_nonconvex2}. SGD with $n=100$ has the least sample complexity but with too limited performance to consider. Furthermore, SGD with either $n=100$ or $n=1000$ has a limiting error neighborhood that prevents its convergence to the exact optimal solution.

Table \ref{table4} then summarizes the number of samples required for the target loss $0.06$ and corresponding relative sample efficiencies for four algorithms. Similarly as the convex case, TSA shows the best relative sample efficiency (requires least samples), but achieves a better performance than SGDs with $n=100$ and $n=1000$ and a comparable performance to SGD with exponential batch-size.

In conclusion, with numerical results for both convex and non-convex problems, the proposed TSA scheme exhibits a fast convergence rate with reduced sample complexity, which reaches a favorable balance among standard SGD algorithms. More importantly, it gives a guideline how to tune SGD parameters appropriately with no need to preset parameters at the outset. The latter may be sensitive and difficult in practice.  

\begin{table} 
\begin{center}  
\caption{Samples required for loss $0.06$ for SGD with $n=20$, non-convex TSA, SGD with $n = 200$, and SGD with $n_k=2^k$. Relative sample efficiency computed as the ratio of samples required with respect to SGD with $n_k=2^k$.}  
\label{table4}
\begin{tabular}{|l|l|l|l| p{2cm}|}  
\hline  
Target loss: 0.06 & Required samples & Relative sample efficiency \\ \hline  
SGD with $n=100$ & $\infty$ & 0 \\ \hline  
TSA &  271687 & 31 \\  \hline
SGD with $n=1000$ & 516000 & 16 \\ \hline
SGD with $n_k = 2^k$ & 8360767 & 1 \\
\hline  
\end{tabular}  
\end{center}  \vspace{-4mm}
\end{table}


\section{Conclusions} \label{sec:conclusion}

This paper investigates stochastic optimization problems that are of critical importance in wide science and engineering areas. The two scale adaptive (TSA) scheme is developed for both convex and non-convex problems, by co-considering the batch-size and the step-size of SGD simultaneously. In particular, the optimal step-size is selected to acquire theoretically largest learning rate, while the batch-size is increased adaptively to tighten the limiting error neighborhood. Equipped with the exact convergence, TSA exhibits the fast rate due to the selected optimal step-size. In the meantime, it only increases the batch-size when necessary, which reduces the sample complexity as much as possible. Numerical experiments are performed to show significant performance of TSA, which well balances rates and variance among standard SGD algorithms.


\appendices 




\section{Proof of Proposition 1} \label{pr:prop1}

\begin{proof}

From the truncated Taylor's expansion for $\mathbb{E}[F(\bbx_{k+1})]$ at $\bbx_{k}$ and the Lipschitz continuity in Assumption \ref{asm1}, we have the inequality
\begin{equation} \label{prprop11}
\begin{split}
F(\bbx_{k+1})&\le F(\bbx_k) \!+\! \nabla F(\bbx_k)^\top\! \big(\bbx_{k+1}-\bbx_k\big)\!+\!\!\frac{L}{2} \| \bbx_{k+1}\!-\!\bbx_k \|^2_2.
\end{split}
\end{equation}
By substituting the update rule \eqref{eq:sgd} of SGD into \eqref{prprop11}, we get
\begin{equation} \label{prprop12}
\begin{split}
F(\bbx_{k+1})&\!\le\! F(\!\bbx_k\!) \!-\! \alpha_k \nabla F(\bbx_k)^\top \nabla\! f_{S_k}(\bbx_k)\!+\!\frac{\alpha_k^2 L}{2}\| \nabla f_{S_k}(\bbx_k) \|^2_2.
\end{split}
\end{equation}
Take the expectation for both sides of \eqref{prprop12}, and by using the fact $\mathbb{E}\left[ \nabla f_{S_k}(\bbx_k)\right] = \mathbb{E}\left[\nabla F(\bbx_k)\right]$, we have
\begin{align}\label{prprop13}
&\mathbb{E}\!\left[ F(\bbx_{k+1})\right] \le \mathbb{E}\left[ F(\bbx_{k})\right] -\alpha_k \mathbb{E}\left[\| \nabla F(\bbx_k)\|^2_2 \right]+\frac{\alpha_k^2 L}{2}\mathbb{E}\left[\| \nabla\! f_{S_k}(\bbx_k)\|^2_2 \right]
\end{align}
where the linearity of expectation is used. By subtracting $F(\bbx^*)$ in both sides of \eqref{prprop13}, we get
\begin{align} \label{prprop135}
&\mathbb{E}\left[ F(\bbx_{k+1})\!-\!F(\bbx^*)\right] \le \mathbb{E}\left[ F(\bbx_{k})\!-\!F(\bbx^*)\right]-\alpha_k \mathbb{E}\!\left[\| \nabla\! F(\bbx_k)\|^2_2 \right]\!+\frac{\alpha_k^2 L}{2}\mathbb{E}\left[\| \nabla f_{S_k}(\bbx_k)\|^2_2 \right].
\end{align}

Consider the third term $\mathbb{E}\left[\parallel \nabla f_{S_k}(\bbx_k)\parallel^2\right]$ in the bound of \eqref{prprop135}. Note that for any random variable $x$, its variance is
\begin{equation}\label{prprop15}
\begin{split}
{\rm Var}[x]  = \mathbb{E}[x^2] - \mathbb{E}[x] \mathbb{E}[x].
\end{split}
\end{equation}
We then have
\begin{align} \label{prprop16}
\mathbb{E}\left[\| \nabla f_{S_k}(\bbx_k)\|^2_2\right] &= \mathbb{E}\!\left[\| \nabla\! F(\bbx_k) \|^2_2\right] \!+\!\mathbb{E}\!\left[\| \nabla\! f_{S_k}\!(\bbx_k)\!-\!\nabla\! F(\bbx_k)\|^2_2\right] = \mathbb{E}\!\left[\| \nabla\! F(\bbx_k) \|^2_2\right] \!+\! \| {\rm Var}\! \left[ \nabla\! f_{S_k}\!(\bbx_k)\right] \|_1.
\end{align}
From \cite{Freund1962}, we proceed to estimate the variance vector ${\rm Var} \left[\nabla f_{S_k}(\bbx_k) \right]$ by
\begin{equation} \label{prprop17}
\begin{split}
{\rm Var}[\nabla f_{S_k}(\bbx_k)] = \frac{{\rm Var} \left[\nabla f_i (\bbx_k) \right]}{\vert S_k \vert} \cdot \frac{N-\vert S_k \vert}{N-1}
\end{split}
\end{equation}
where ${\rm Var} \left[\nabla f_i (\bbx_k) \right]$ is the population variance vector and $N$ is the number of samples drawn to approximate the distribution $\bbp$, which should be infinity in the stochastic optimization problem \eqref{eq:main_prob}. By substituting \eqref{prprop17} into $\| {\rm Var} \left[ \nabla f_{S_k}(\bbx_k)\right] \|_1$ with $\vert S_k \vert = n_k$, we have
\begin{equation} \label{prprop18}
\begin{split}
\| {\rm Var} \left[ \nabla f_{S_k}(\bbx_k)\right] \|_1 \le \| \frac{{\rm Var}[\nabla f_i(\bbx_k)]}{n_k} \|_1 \le \frac{w}{n_k}
\end{split}
\end{equation}
where the last inequality is due to Assumption \ref{as2}. By substituting \eqref{prprop18} into \eqref{prprop16} and then into \eqref{prprop135}, we get
\begin{align} \label{prprop19}
&\mathbb{E}\left[ F(\bbx_{k+1})-F(\bbx^*)\right] \le \mathbb{E}\!\left[ F(\bbx_{k})\!-\!F(\bbx^*)\right] \!-\! \left(\! \alpha_k \!-\! \frac{L\alpha_k^2}{2}\! \right)\! \mathbb{E}\!\left[\| \nabla\! F(\bbx_k) \|^2_2\right] \!+\! \frac{\alpha_k^2 L w}{2n_k}.
\end{align}

Now consider the term $\| \nabla F(\bbx_k)\|^2_2$. From Assumption \ref{as3}, it holds that for any $ \bbx,\bby \in \mathbb{R}^p$,
\begin{equation} \label{prprop14}
\begin{split}
F(\bbx)\ge F(\bby)+\nabla F(\bby)^\top (\bbx-\bby) + \frac{\ell}{2}\| \bbx-\bby \|^2_2.
\end{split}
\end{equation}
Since the right side of \eqref{prprop14} is a quadratic function with minimal value at $\widehat{\bbx}=\bby-\frac{1}{\ell}\nabla F(\bby)$, we have
\begin{equation}
\begin{split}
F(\bbx)&\ge F(\bby)+\nabla\! F(\bby)^\top (\widehat{\bbx}-\bby) + \frac{\ell}{2}\| \widehat{\bbx}-\bby\|^2_2 = F(\bby) - \frac{1}{2\ell}\| \nabla F(\bby) \|^2_2.
\end{split}
\end{equation}
Then let $\bbx=\bbx^*$ and $\bby=\bbx_k$, and we get
\begin{equation} \label{prprop145}
\begin{split}
\| \nabla F(\bbx_k) \|^2_2 \ge 2\ell\left(F(\bbx_k)-F(\bbx^*)\right).
\end{split}
\end{equation}

By substituting \eqref{prprop145} into \eqref{prprop19}, we obtain
\begin{equation} \label{prprop110}
\begin{split}
&\mathbb{E}\left[ F(\bbx_{k+1})-F(\bbx^*)\right] \le \left( 1-2\alpha_k \ell + L\ell\alpha_k^2 \right) \mathbb{E}\left[ F(\bbx_{k})-F(\bbx^*)\right] + \frac{\alpha_k^2 L w}{2n_k}.
\end{split}
\end{equation}
Observe that \eqref{prprop110} is a recursion process such that we can continue deriving the bound until it is represented by the initial condition $F(\bbx_{0})-F(\bbx^*)$ as
\begin{align}
\mathbb{E}\!\left[ F(\bbx_{k+1})\!-\!F(\bbx^*)\right] &\!\le\! \big( \prod_{i
=0}^k\! 1\!\!-\!2\alpha_i \ell \!+\! L\ell\alpha_i^2 \big) \big( F(\bbx_{0})\!-\!F(\bbx^*)\big) + \sum_{i=0}^k\! \big( \frac{\alpha_i^2 L w}{2n_i}\! \prod_{j=i+1}^k\! 1\!-\!2\alpha_j \ell \!+\! L\ell\alpha_j^2 \big).
\end{align}
Let $r(\alpha) = 1-2\alpha \ell + L\ell\alpha^2$ and we complete the proof.

\end{proof}


\section{Proof of Corollary 1} \label{pr:coro11}

\begin{proof}
By substituting $\alpha_k = \alpha$ and $n_k = n$ into Proposition \ref{prop1}, we have
\begin{equation}\label{pr:coroeq1}
\begin{split}
 &\mathbb{E}\left[ F(\bbx_{k+1})-F(\bbx^*)\right] \le r(\alpha)^{k+1} \left( F(\bbx_{0})\!-\!F(\bbx^*) \right) \!+\! \sum_{i=0}^k \left( \frac{\alpha^2 L w}{2n} r(\alpha)^{k-i} \right).
\end{split}
\end{equation}
Note that the second term in \eqref{pr:coroeq1} is a geometric series with the common ratio $r(\alpha)$. By summing terms up, we have
 \begin{equation}\label{pr:coroeq2}
\begin{split}
\sum_{i=0}^k \left( \frac{\alpha^2 L w}{2n} r(\alpha)^{k-i} \right) &= \frac{\alpha^2 L w}{2n} \cdot \frac{1-r(\alpha)^{k+1}}{1-r(\alpha)} \le \frac{\alpha^2 L w}{2n (1-r(\alpha))},
\end{split}
\end{equation}
where $r(\alpha) < 1$ is used in the last inequality. By substituting $r(\alpha) = 1-2\alpha \ell + L\ell\alpha^2$ into \eqref{pr:coroeq2}, we complete the proof.
\end{proof}


\section{Proof of Proposition 2} \label{pr:prop2}

\begin{proof}
From the truncated Taylor's expansion for $\mathbb{E}[F(\bbx_{k+1})]$ at $\bbx_{k}$ and the Lipschitz continuity in Assumption \ref{asm1}, we have
\begin{equation} \label{eq:prprop22}
\begin{split}
\mathbb{E}\left[ F(\bbx_{k+1}) \right] & \le \mathbb{E}\big[ F(\bbx_{k})+ \nabla F(\bbx_{k})^\top (\bbx_{k+1}-\bbx_k) + \frac{L}{2} \| \bbx_{k+1}-\bbx_k \|_2^2 \big].
\end{split}
\end{equation}
Substituting the update rule \eqref{eq:sgd} of SGD into \eqref{eq:prprop22}, we get
\begin{equation}\label{eq:prprop225}
\begin{split}
\mathbb{E}\left[ F(\bbx_{k+1})\right] &\le \mathbb{E}\big[ F(\bbx_{k}) - \alpha_k \nabla F(\bbx_{k})^\top \nabla f_{S_k}(\bbx_k) + \frac{\alpha_k^2 L}{2} \| \nabla f_{S_k}(\bbx_k) \|_2^2\big].
\end{split}
\end{equation}
With the linearity of expectation and the fact that $\mathbb{E}[\nabla f_{S_k}(\bbx_k)]=\mathbb{E}[\nabla F(\bbx_k)]$, \eqref{eq:prprop225} becomes
\begin{align}\label{eq:prprop23}
&\mathbb{E}\left[ F(\bbx_{k+1})\right] \le \mathbb{E}\left[ F(\bbx_{k}) \right] - \alpha_k \mathbb{E}\left[ \| \nabla F(\bbx_{k})\|_2^2 \right] +\frac{\alpha_k^2 L}{2} \mathbb{E}\left[ \| \nabla f_{S_k}(\bbx_k) \|_2^2\right].
\end{align}
Consider the term $\mathbb{E}\left[\| \nabla f_{S_k}(\bbx_k)\|_2^2\right]$ in the bound of \eqref{eq:prprop23}. By substituting \eqref{prprop18} in the proof of Proposition \ref{prop1} into \eqref{eq:prprop23}, we get
\begin{align}\label{eq:prprop27}
&\mathbb{E}\left[ F(\bbx_{k+1})\right] \le \mathbb{E}\left[ F(\bbx_{k})\right]-\big(\alpha_k-\frac{\alpha_k^2 L}{2}\big) \mathbb{E}\left[ \| \nabla F(\bbx_k)\|^2 \right] + \frac{\alpha_k^2 L w}{2 n_k}.
\end{align}

Here, we use $\parallel \nabla_\bbx F(\bbx_{t}) \parallel^2_2 \le \epsilon$ as the convergence criterion to judge the approximate stationary in non-convex optimization problems, where $\epsilon$ can be any small value. We then focus on the term $\mathbb{E}\left[ \| \nabla F(\bbx_{t})\|_2^2 \right]$ in the bound of \eqref{eq:prprop27} and move it to the left side as
\begin{align} \label{eq:prprop28}
&\mathbb{E}\left[ \| \nabla F(\bbx_{k})\|_2^2 \right] \le \frac{1}{\alpha_k-\frac{\alpha_k^2 L}{2}} \mathbb{E}[ F(\bbx_{k}) - F(\bbx_{k+1}) ]+ \frac{\alpha_k^2 L w}{2n_k \big( \alpha_k-\frac{\alpha_k^2 L}{2}\big)}.
\end{align}
The bound in \eqref{eq:prprop28} cannot be used to show the convergence due to the existence of term $\mathbb{E}[ F(\bbx_{k}) - F(\bbx_{k+1}) ]$. To handle this issue, note that \eqref{eq:prprop28} holds for all iterations $k=0,1,\ldots$ and we have $\alpha_k = \alpha$ as a constant, such that we have
\begin{align}
&\!\!\sum_{i=0}^{k} \mathbb{E}\!\left[ \| \nabla F(\bbx_{i})\|_2^2 \right] \!\le\! \frac{1}{\alpha\!-\frac{\alpha^2 L}{2}} \mathbb{E}[ F(\bbx_{0}) \!-\! F(\bbx_{k+1}) ]\!+\! \sum_{i=0}^{k}\! \frac{ \alpha^2 L w}{2n_i \big(\alpha-\frac{\alpha^2 L}{2}\big)}.
\end{align}
Thus, we have
\begin{align} \label{eq:prprop29}
&\min_{0 \le i \le {k}} \mathbb{E}\left[ \| \nabla F(\bbx_{i})\|_2^2 \right] \le \frac{1}{k+1} \sum_{i=0}^{k} \mathbb{E}\left[ \| \nabla F(\bbx_{i})\|_2^2 \right] \le \frac{1}{\alpha k-\frac{\alpha^2 L k}{2}} \left( F(\bbx_{0}) - F(\bbx^*) \right)+ \sum_{i=0}^{k} \frac{ \alpha^2 L w}{2k n_i \left(\alpha-\frac{\alpha^2 L}{2}\right)}
\end{align}
where the last inequality is due to the fact that $F(\bbx^*)=\min_{\bbx} F(\bbx) \le F(\bbx_{k+1})$ with $\bbx^*$ an optimal solution of \eqref{eq:main_prob}. We then complete the proof.
\end{proof}



\section{Proof of Theorem 1} \label{pr:thm1}

\begin{proof}
\textbf{The post TSA scheme.} Consider iteration $k$ at $t$-th inner time-scale. From Proposition \ref{prop1}, we have
\begin{equation} \label{prthm11}
\begin{split}
&\mathbb{E}\left[ F(\bbx_{k})-F(\bbx^*)\right]\le \big(1-\frac{\ell}{L} \big)^{k-K}\mathbb{E}\left[ F(\bbx_{K})-F(\bbx^*)\right] + \frac{w}{2 n_t\ell}
\end{split}
\end{equation}
with $K = \sum_{i=0}^{t-1} K_i$ and $K_i$ the duration of $i$-th inner time-scale. From the stop criterion \eqref{eq:innercond05}, we have 
\begin{align} \label{prthm1151}
\mathbb{E}\left[ F(\bbx_{K})-F(\bbx^*)\right] \le 2^{t}\big(1-\frac{\ell}{L} \big)^{K} \left( F(\bbx_0)-F(\bbx^*)\right)+\frac{w}{2 n_{t-1}\ell}.
\end{align}
Now multiplying $1-\frac{\ell}{L}$ on both sides of \eqref{prthm1151} and using the stop criterion \eqref{eq:innercond05} again, we get
\begin{equation} \label{prthm115}
\begin{split}
\big(1-\frac{\ell}{L}\big)\mathbb{E}\left[ F(\bbx_{K})-F(\bbx^*)\right] \le \frac{w}{2 n_{t-1}\ell}+ \frac{w}{2 n_{t-1}\ell}.
\end{split}
\end{equation}
Now substituting \eqref{prthm115} into \eqref{prthm11}, we have
\begin{equation} \label{prthm12}
\begin{split}
&\mathbb{E}\left[ F(\bbx_{k})-F(\bbx^*)\right] \le 2\big(1-\frac{\ell}{L} \big)^{k-K-1}\frac{w}{2 n_{t-1}\ell} + \frac{w}{2 n_t\ell}.
\end{split}
\end{equation}
The bound in \eqref{prthm12} comprises two terms. For the first term, observe that at each inner time-scale the convergence rate term $Q_1^t$ keeps decreasing while the error neighborhood term $Q_2^t$ remains constant by definition, such that the duration $K_t$ is finite. Then $t \to \infty$ as $k \to \infty$. Therefore, $\lim_{k \to \infty} n_{t-1} = \lim_{t \to \infty} n_{t-1} = \infty$ and $\lim_{k \to \infty}w/({2 n_{t-1}\ell})=0$. In addition provided that $1-\ell /L \le 1$ and $k-K-1\ge 0$, we have
\begin{equation}\label{prthm13}
\begin{split}
\lim_{k \to \infty} 2\big(1-\frac{\ell}{L} \big)^{k-\sum_{i=0}^{t-1}K_{i}-1} \frac{w}{2n_{t-1}\ell } = 0.
\end{split}
\end{equation}

For the second term, $\lim_{k \to \infty} n_{t} = \infty$ since $n_t > n_{t-1}$ and thus $\lim_{k \to \infty}w/({2 n_t \ell})=0$. By substituting this result and \eqref{prthm13} into \eqref{prthm12}, we get $\lim_{k\to \infty}\mathbb{E}\left[ F(\bbx_{k})-F(\bbx^*)\right]=0$.

With the strong convexity from Assumption \ref{as3} and the fact that the gradient of optimal solution $\bbx^*$ is the null vector, we have $F(\bbx_k) - F(\bbx^*) \ge (\ell/2) \| \bbx_k - \bbx^* \|^2_2$. By using this result, we get $\lim_{k \to \infty}\! \mathbb{E}\left[ \| \bbx_k - \bbx^* \|_2 \right]\! =\! 0$.

\textbf{The prior TSA scheme.} Based on the stop criterions \eqref{eq:innercond05} and \eqref{eq:innercond1} of the post and the prior TSA schemes, the prior increases the batch-size faster than the post. The prior then has a faster rate and thus converges exactly as well.

\end{proof}


\section{Proof of Theorem 2} \label{pr:prop3}

\begin{proof}
\textbf{The post TSA scheme.} Consider iteration $k$ at $t$-th inner time-scale. From the stop criterion \eqref{eq:innercond05}, we have
\begin{align} \label{prprop315}
\mathbb{E}\left[ F(\bbx_{k})-F(\bbx^*)\right] &\le 2^{t}\big(1-\frac{\ell}{L} \big)^{k} \left( F(\bbx_0)-F(\bbx^*)\right)+\frac{w}{2 n_t\ell} \le 2^{t+1}\!\big(1-\frac{\ell}{L} \!\big)^{k} \left( F(\bbx_0)-F(\bbx^*)\right).
\end{align}
The rate of post TSA is approximately $\ccalO(2^t(1-\ell/L)^k)$. 

Assume TSA uses the multiplicative rule \eqref{eq:multi} for augmenting the batch-size. In this case, $K_t \le \left\lceil \log_{1-\frac{\ell}{L}} \frac{1}{2m} \right\rceil $ for all $t \ne 0$ according to \eqref{eq:innercond05} where $\left\lceil \cdot \right\rceil$ is the ceil function. By using this result and the fact $k=\sum_{i=0}^{t-1}K_i+k_t$, we get
\begin{equation} \label{prprop33}
\begin{split}
t \ge \frac{k-K_0}{\left\lceil \log_{1-\frac{\ell}{L}} \frac{1}{2m} \right\rceil}.
\end{split}
\end{equation}
We can also refer from the stop criterion \eqref{eq:innercond05} that
\begin{align}\label{prprop341}
2^{t}\big(1-\frac{\ell}{L} \!\big)^{k} \left( F(\bbx_0)-F(\bbx^*)\right) \le \frac{w}{2 n_{t-1}\ell},~\frac{w}{2 n_t\ell} \le \frac{w}{2 n_{t-1}\ell}.
\end{align}
By substituting \eqref{prprop33} and \eqref{prprop341} into the first inequality of \eqref{prprop315} and using the fact $n_{t-1}=n_0 m^{t-1}$, we get
\begin{align} \label{prprop35}
\mathbb{E}\left[ F(\bbx_k) \!-\! F(\bbx^*) \right] &\le \frac{w}{n_0 m^{t-1} \ell} \le \frac{w}{n_0 \ell} \big(\frac{1}{m}\big)^{\frac{k}{\left\lceil \log_{1-\frac{\ell}{L}} \frac{1}{2m} \right\rceil}- \frac{K_0}{\left\lceil \log_{1-\frac{\ell}{L}} \frac{1}{2m} \right\rceil}-1}.
\end{align}
Therefore, the rate is approximately $\ccalO((1/m)^{k/\log_{1-\ell/L}\frac{1}{2m}})$.

\textbf{The prior TSA scheme.} Consider iteration $k$ at $t$-th inner time-scale. By substituting \eqref{eq:error_conditions1} into \eqref{eq:innercond1}, the stop criterion of prior TSA at $t$-th inner time-scale is 
\begin{equation} \label{prprop36}
\begin{split}
K_t =  \max_{k_t} \left\{ \big( 1 - \frac{\ell}{L} \big)^{K + k_{t}}\! \left(F(\bbx_{0})-F(\bbx^*)\right) \ge \frac{w}{2 n_t\ell} \right\}
\end{split}
\end{equation}
with $K=\sum_{i=0}^{t-1}K_i$. From \eqref{eq:decrement} in Proposition \ref{prop1}, we have
\begin{align} \label{prprop37}
\mathbb{E}\left[ F(\bbx_k) - F(\bbx^*) \right] &\le \big( 1 - \frac{\ell}{L} \big)^{k} \left(F(\bbx_{0})-F(\bbx^*)\right) + \sum_{i=0}^{t-1} \big( 1 - \frac{\ell}{L} \big)^{\sum_{j=i+1}^{t}K_j+k_t} \frac{w}{2 n_i\ell}+\frac{w}{2 n_{t}\ell}\nonumber\\
& \le (t+2) \big( 1 - \frac{\ell}{L} \big)^{k} \left(F(\bbx_{0})-F(\bbx^*)\right)
\end{align}
where the last inequality is due to the stop criterion \eqref{prprop36} that applies for each inner time-scale. As such, the convergence rate of prior TSA is approximately $\ccalO(t(1-\ell/L)^k)$. 

Assume TSA uses the multiplicative rule \eqref{eq:multi} for augmenting the batch-size. Similarly as \eqref{prprop33}, we have
\begin{equation} \label{prprop38}
\begin{split}
t \le \frac{k-K_0}{\left\lfloor \log_{1-\frac{\ell}{L}} \frac{1}{m} \right\rfloor}+1
\end{split}
\end{equation}
with $\left\lfloor \cdot \right\rfloor$ the floor function. Substituting \eqref{prprop38} in \eqref{prprop37}, we get
\begin{align} \label{prprop39}
&\mathbb{E}\left[ F(\bbx_k) - F(\bbx^*) \right] \le \big(\frac{k}{\left\lfloor \log_{1-\frac{\ell}{L}} \frac{1}{m} \!\right\rfloor} - \frac{K_0}{\left\lfloor\! \log_{1-\frac{\ell}{L}} \frac{1}{m} \right\rfloor}+3\big) \big( 1 - \frac{\ell}{L} \big)^{k} \left(F(\bbx_{0})-F(\bbx^*)\right).
\end{align}
The rate is approximately $\ccalO((1-\ell/L)^kk/\log_{1-\ell/L}\frac{1}{m})$.
\end{proof}


\section{Proof of Theorem 3} \label{pr:thm2}

\begin{proof}
An $\epsilon$-suboptimal solution is a solution $\bbx_{k}$ that satisfies $F(\bbx_{k}) - F(\bbx^*) \le \epsilon$. From Proposition \ref{prop1}, for SGD with constant batch-size $n$ and step-size $\alpha = 1/L$ to guarantee an $\epsilon$-suboptimal solution, we require
\begin{equation} \label{prthm21}
\begin{split}
F(\bbx_{k}) - F(\bbx^*) \le \big( 1-\frac{\ell}{L} \big)^{k} D + \frac{w}{2 n\ell} \le \epsilon. 
\end{split}
\end{equation}
Since $w/({2 n\ell})$ is constant and $\left( 1-\ell/L \right)^{k} D$ keeps decreasing, it is reasonable to stop the iteration when $\left( 1-\ell/L \right)^{k} D \le w/({2 n\ell})$. Based on this consideration, \eqref{prthm21} is equivalent to
\begin{equation}\label{prthm22}
\big( 1-\frac{\ell}{L} \big)^{k} D \le \frac{\epsilon}{2},~~\frac{w}{2 n \ell} \le \frac{\epsilon}{2}.
\end{equation}
From \eqref{prthm22}, we obtain $k \ge \left\lceil \log_{1-\frac{\ell}{L}}\frac{\epsilon}{2D} \right\rceil$ and $n \ge \left\lceil w/({\ell\epsilon})\right\rceil$. Thus, we stop the iteration at $k = \left\lceil \log_{1-\frac{\ell}{L}}\frac{\epsilon}{2D} \right\rceil$ and the total number of training samples required for SGD is
\begin{equation}\label{prthm23}
N_{SGD} = \left\lceil \log_{1-\frac{\ell}{L}}\frac{\epsilon}{2D} \right\rceil n . 
\end{equation}

Consider TSA with initial batch-size $n_0=1$. Assume it achieves the $\epsilon$-suboptimal solution at $t$-th inner time-scale. 

\textbf{The post TSA scheme.} From \eqref{prprop315}, to achieve $\epsilon$-suboptimality for the post TSA, we require
\begin{equation} \label{prthm24}
\begin{split}
&\mathbb{E}\left[ F(\bbx_{k})-F(\bbx^*)\right] \le 2^t \big( 1-\frac{\ell}{L} \big)^{k} D + \frac{w}{2 n_t \ell } \le \epsilon.
\end{split}
\end{equation}
With the same consideration as in \eqref{prthm22}, \eqref{prthm24} is equivalent to
\begin{align} \label{prthm25}
2^t \big( 1-\frac{\ell}{L} \big)^{k} D \le \frac{\epsilon}{2},~\frac{w}{2 n_t \ell} \le \frac{\epsilon}{2}.
\end{align}
Assume the second inequality in \eqref{prthm25} is satisfied at $t$-th inner time-scale. Then the first inequality will be satisfied at the end of $t$-th inner time-scale based on \eqref{eq:innercond05}. In particular, TSA first goes through $t$ inner/outer time-scales, each of which contains $K_i$ iterations with each iteration requiring $n_0m^i $ samples for $i=0,...,t-1$. It then runs $K_t+1$ iterations at $t$-th inner time-scale and each iteration uses $n_t$ samples. Here, note that TSA does not need to step into $t$-th outer time-scale to further increase batch-size since one more iteration at $t$-th inner time-scale is enough to obtain the target accuracy according to \eqref{eq:innercond05}. In addition with the multiplicative rule \eqref{eq:multi} and the stop criterion \eqref{eq:innercond05}, the duration $K_t$ is bounded by
\begin{equation} \label{prthm235}
K_t \le \left\lceil \log_{1-\frac{\ell}{L}} \frac{1}{2m} \right\rceil = \widehat{K}, ~\forall~ t=1,2,\ldots. .
\end{equation}
The total number of training samples required for TSA is
\begin{align}\label{prthm265}
N_{TSA} &= \sum_{i=0}^{t-1} K_i n_i + (K_t+1)n_t \le (K_0-\widehat{K}) n_0 + \widehat{K} \sum_{i=0}^{t} m^i n_0+ m^t n_0 \nonumber\\
& \le (K_0-\widehat{K}) n_0 + \widehat{K} \frac{m^{t+1} n_0 -n_0}{m-1}+ m^t n_0
\end{align}
where $n_t = n_0 m^t$ is used. By substituting this result and the fact $n_0=1$ into \eqref{prthm265}, we have 
\begin{equation} \label{prthm27}
\begin{split}
N_{TSA} &\le  K_0-\widehat{K} + \big(\frac{m}{m-1}\widehat{K}+1\big) n_t
\end{split}
\end{equation}
with $n_t \ge \left\lceil w/({\ell\epsilon})\right\rceil$ from \eqref{prthm25}. 

Note that both $n$ of SGD and $n_t$ of TSA need to be larger than or equal to $\left\lceil w/({\ell\epsilon})\right\rceil$. Without loss of generality, let $n=n_t = \left\lceil w/({\ell\epsilon})\right\rceil$ for a clear comparison. By substituting \eqref{prthm235} into \eqref{prthm27} and comparing the latter with \eqref{prthm23}, the ratio $\gamma$ is bounded by
\begin{equation}\label{prthm29}
\begin{split}
\gamma & \le \frac{\frac{m}{m-1} \left\lceil \log_{1-\frac{\ell}{L}} \frac{1}{2m} \right\rceil +1 }{\left\lceil \log_{1-\frac{\ell}{L}}\frac{\epsilon}{2D} \right\rceil} + \frac{K_0-\left\lceil \log_{1-\frac{\ell}{L}} \frac{1}{2m} \right\rceil}{\left\lceil \log_{1-\frac{\ell}{L}}\frac{\epsilon}{2D} \right\rceil \left\lceil \frac{w}{\ell \epsilon} \right\rceil} \le \frac{ \frac{m}{m-1} \left\lceil \log_{1-\frac{\ell}{L}} \frac{1}{2m} \right\rceil+1}{  \left\lceil \log_{1-\frac{\ell}{L}}\frac{\epsilon}{2D} \right\rceil} + \mathcal{O}(\epsilon).
\end{split}
\end{equation}

\textbf{The prior TSA scheme.} Recall \eqref{prprop37} and to achieve an $\epsilon$-suboptimal solution for the prior TSA, we require
\begin{align} \label{prthm211}
&\mathbb{E}\left[ F(\bbx_{k}) - F(\bbx^*) \right] \le (t+1) \big( 1 - \frac{\ell}{L} \big)^{k} D + \frac{w}{2 n_t \ell} \le \epsilon
\end{align}
which is equivalent to require
\begin{equation}\label{prthm212}
(t+1) \big( 1 - \frac{\ell}{L} \big)^{k} D \le \frac{\epsilon}{2},~~\frac{w}{2 n_t \ell} \le \frac{\epsilon}{2}.
\end{equation}
Assume the second inequality in \eqref{prthm212} is satisfied at $t$-th inner time-scale. For an $\epsilon$-suboptimal solution, we first let TSA go through $t$ inner/outer time-scales, and then stay at $t$-th inner time-scale and perform $k_t$ iterations until the first inequality in \eqref{prthm212} is satisfied. In addition from \eqref{prprop36}, the duration $K_t$ in this case satisfies $K_t \le \left\lceil \log_{1-\frac{\ell}{L}} \frac{1}{m} \right\rceil = \widehat{K}$. We then follow \eqref{prthm265} and bound the total number of training sample as
\begin{equation}\label{prthm213}
\begin{split}
N_{TSA} &\le (K_0-\widehat{K}) n_0 + \widehat{K} \frac{m^{t} n_0 -n_0}{m-1}+k_t n_0 m^t.
\end{split}
\end{equation}
where $n_t=n_0m^t$ is used. Similarly, without loss of generality, let $n=n_t = \left\lceil w/({\ell\epsilon})\right\rceil$ for a clear comparison. With $n_0=1$, we then obtain $t = \log_{m}\! \left\lceil\!w/(\ell\epsilon)\! \right\rceil$. To satisfy the first inequality in \eqref{prthm212}, we require
\begin{equation}\label{prthm2135}
\begin{split}
k_t \le \widehat{K} +1 + \left\lceil\log_{1 -\frac{\ell}{L}} \frac{1}{\log_{m} \left\lceil\frac{w}{\ell\epsilon} \right\rceil+1} \right\rceil.
\end{split}
\end{equation}
By substituting \eqref{prthm2135} into \eqref{prthm213}, we have
\begin{align} \label{prthm214}
&N_{TSA}\le  K_0-\widehat{K} + \big( \frac{m\widehat{K}}{m-1}+1+ \left\lceil\log_{1 -\frac{\ell}{L}} \frac{1}{\log_{m} \left\lceil\frac{w}{\ell\epsilon} \right\rceil+1} \right\rceil\big)\left\lceil \frac{w}{\ell \epsilon} \right\rceil.
\end{align}
Then by comparing \eqref{prthm23} with \eqref{prthm214}, we get
\begin{equation}\label{prthm215}
\begin{split}
\gamma & \le \frac{ \frac{m}{m-1} \left\lceil \log_{1-\frac{\ell}{L}} \frac{1}{m} \right\rceil + \left\lceil\log_{1 -\frac{\ell}{L}} \frac{1}{\log_{m}\left\lceil \frac{w}{\ell\epsilon} \right\rceil+1} \right\rceil+1}{ \left \lceil \log_{1-\frac{\ell}{L}}\frac{\epsilon}{2D} \right\rceil} + \mathcal{O}(\epsilon)
\end{split}
\end{equation}
completing the proof.
\end{proof}


\section{Proof of Corollary 2} \label{pr:coro2}

\begin{proof}
By substituting $m=2$ into \eqref{prthm29} and using the fact $2\left\lceil \log_{1-\frac{\ell}{L}} \frac{1}{4} \right\rceil+1 \le \left\lceil \log_{1-\frac{\ell}{L}} \frac{1}{16} \right\rceil+2$, we get
\begin{equation}\label{prcoro13}
\begin{split}
\gamma &\le \frac{\left\lceil \log_{1-\frac{\ell}{L}} \frac{1}{16} \right\rceil+2}{\left\lceil \log_{1-\frac{\ell}{L}}\frac{\epsilon}{2D} \right\rceil} + \mathcal{O}(\epsilon) = \frac{\left\lceil \log_{1-\frac{\ell}{L}} \frac{(L-\ell)^2}{16L^2} \right\rceil}{\left\lceil \log_{1-\frac{\ell}{L}}\frac{\epsilon}{2D} \right\rceil} + \mathcal{O}(\epsilon)
\end{split}
\end{equation}
completing the proof.
\end{proof}


\section{Proof of Theorem 4} \label{pr:thm3}

\begin{proof}
\textbf{The post TSA scheme.} Consider iteration $k$ at $t$-th inner time-scale. From \eqref{nonconvex1}, we have
\begin{equation} \label{prthm31}
\begin{split}
\min_{0 \le i \le {k-1}} \mathbb{E} \left[ \| \nabla F(\bbx_{i})\|_2^2 \right]&\le \frac{2L}{k} \left( F(\bbx_0)-F(\bbx^*) \right) + \sum_{i=0}^{t-1}\frac{K_i w}{k n_i} + \frac{ w}{n_t}.
\end{split}
\end{equation}
By extracting the factor $\sum_{j=0}^{t-1}K_j+1$, \eqref{prthm31} becomes
\begin{align} \label{prthm32}
\min_{0 \le i \le {k-1}} \mathbb{E}\left[ \| \nabla F(\bbx_{i})\|_2^2 \right] &\le \frac{\sum_{j=0}^{t-1}K_j+1}{k} \left( \frac{2L\big( F(\bbx_0)-F(\bbx^*) \big)}{\sum_{j=0}^{t-1}K_j+1} + \sum_{i=0}^{t-2}\frac{K_i}{\sum_{j=0}^{t-1}K_j+1} \frac{w}{n_i} + \frac{K_{t-1}}{\sum_{j=0}^{t-1}K_j+1} \frac{w}{n_{t-1}} \right) + \frac{w}{n_t} \nonumber\\
& \le \frac{2 \big( \sum_{j=0}^{t-1}K_j+1\big)}{k} \frac{w}{n_{t-1}}+ \frac{w}{n_t}.
\end{align}
where the stop criterion \eqref{eq:noninnercond2} and the fact $K_{t-1} \le \sum_{j=0}^{t-1}K_j+1$ is used in the last inequality. The bound of \eqref{prthm32} comprises two terms. For the first term, observe that $Q_1^t$ keeps decreasing while $Q_2^t$ remains constant at each inner time-scale by definition, such that the duration $K_t$ is finite. Then $t \to \infty$ as $k \to \infty$. Therefore, $\lim_{k \to \infty} n_{t-1} = \lim_{t \to \infty} n_{t-1} = \infty$ and $\lim_{k \to \infty}w/n_{t-1}=0$. In addition provided that
\begin{equation}
\begin{split}
\lim_{k \to \infty} \frac{ \sum_{j=0}^{t-1}K_j+1}{k} = \lim_{k \to \infty} \frac{\sum_{j=0}^{t-1}K_j+1}{k_t+\sum_{j=0}^{t-1}K_j} \le 1
\end{split}
\end{equation}
with $k_t \ge 1$, we have
\begin{equation}\label{prthm33}
\begin{split}
\lim_{k \to \infty} \frac{2 \big( \sum_{j=0}^{t-1}K_j+1\big)}{k} \frac{w}{n_{t-1}} = 0.
\end{split}
\end{equation}

For the second term, $\lim_{k \to \infty} n_{t} = \infty$ because $n_t > n_{t-1}$  such that $\lim_{k \to \infty}w/n_t=0$. By substituting this result and \eqref{prthm33} into \eqref{prthm32}, we get $\lim_{k\to \infty}\mathbb{E}\left[ \| \nabla F(\bbx_{k})\|_2^2 \right]=0$.

\textbf{The prior TSA scheme.} By comparing the stop criterions, the prior TSA increases the batch-size faster than the post. As such, the prior converges faster and has exact convergence as well.
\end{proof}


\section{Proof of Theorem 5} \label{pr:prop4}

\begin{proof} We analyze the prior TSA with additive rule \eqref{eq:add} and multiplicative rule \eqref{eq:multi} separately.

\textbf{The additive rule.} Consider iteration $k$ at $t$-th inner time-scale. By substituting \eqref{eq:nonerror_conditions3} into \eqref{eq:noninnercond2}, the stop criterion of non-convex prior TSA at $t$-th inner time-scale is
\begin{equation} \label{eq:prprop41}
\begin{split}
K_t = \max_{k_t} \left\{ \frac{2L}{k_t+\sum_{i=0}^{t-1}K_i} \left( F(\bbx_0)-F(\bbx^*) \right) \ge \frac{w}{n_t} \right\}.
\end{split}
\end{equation}
Let $D = 2L\left( F(\bbx_0)-F(\bbx^*) \right)$ and by using \eqref{nonconvex1} and \eqref{eq:prprop41}, we have
\begin{align} \label{eq:prprop415}
\min_{0 \le i \le {k-1}} \mathbb{E}\left[ \| \nabla F(\bbx_{i})\|_2^2 \right]\!&\le\! \frac{D}{k} \big(1+1+\sum_{i=1}^{t-1}\frac{K_i}{\sum_{j=0}^{i}K_j}+\frac{k_t}{k}\big).
\end{align}
Now consider the duration of inner time-scale $K_t$ in the bound of \eqref{eq:prprop415}. We use the induction to provide lower and upper bounds for it. At $1$-st inner time-scale, by using the stop criterion \eqref{eq:prprop41} and the fact $n_1=n_0+\beta$, we have
\begin{equation} \label{eq:prprop42}
\begin{split}
\left\lfloor \frac{K_0 \beta}{n_0} \right\rfloor \le K_1 \le \left\lceil \frac{(K_0+1) \beta}{n_0} \right\rceil.
\end{split}
\end{equation}
Let $C_L = \lfloor K_0 \beta/n_0 \rfloor$ and $C_U = \lceil (K_0+1) \beta/n_0 \rceil$ be concise notations and assume \eqref{eq:prprop42} holds for $K_{t-1}$. With \eqref{eq:prprop41} and the fact $n_t=n_{t-1}+\beta$, we have
\begin{subequations}\label{eq:prprop43}
\begin{align}
K_t &\ge \left\lfloor \frac{(K_0 + \sum_{i=1}^{t-1}K_i)\beta}{n_0 + (t-1)\beta} \!\right\rfloor \ge \left\lfloor \frac{(n_0 + (t-1)\beta) C_L }{n_0 + (t-1)\beta} \right\rfloor = C_L,\\
K_t &\le \left\lceil \frac{(K_0 + \sum_{i=1}^{t-1}K_i+1)\beta}{n_0 + (t-1)\beta} \right\rceil \le \left\lceil \frac{(n_0 + (t-1)\beta) C_U }{n_0 + (t-1)\beta} \right\rceil = C_U.
\end{align}
\end{subequations}
Therefore, we have $C_L \le K_{t} \le C_U ~\forall~ t=1,2,\ldots$.
By substituting this result into \eqref{eq:prprop415}, we have
\begin{align} \label{eq:prprop45}
\min_{0 \le i \le {k-1}} \mathbb{E}\left[ \| \nabla F(\bbx_{i})\|_2^2 \right]\!& \le \frac{2D}{k} + \frac{DC_U\sum_{i=1}^{t} \frac{1}{i}}{C_Lk}
\end{align}
where the inequality $k_t/k \le K_t/\sum_{i=0}^t K_t$ with $K_t - k_t = \sum_{i=0}^t K_t - k \ge 0$ is used. Since $\sum_{i=1}^{t} 1/i$ is the harmonic series and $K_t \ge C_L$, we have \cite{Atanassov1986}
\begin{align} \label{eq:prprop46}
\sum_{i=1}^{t} \frac{1}{i} = \ln (t+1) + c,~t \le \frac{k-K_0}{C_L}+1
\end{align}
with constant $c$. By substituting \eqref{eq:prprop46} into \eqref{eq:prprop45}, we get
\begin{align} \label{eq:prprop48}
&\min_{0 \le i \le {k-1}} \mathbb{E}\left[ \| \nabla F(\bbx_{i})\|_2^2 \right] \le \frac{\frac{DC_U}{C_L}\ln \big(\frac{k-K_0}{C_L}+2\big) + \frac{DC_Uc}{C_L} c + 2D}{k}.
\end{align}
Therefore, the convergence rate is approximately $\ccalO(\log k / k)$.

\textbf{The multiplicative rule.} Consider iteration $k$ at $t$-th inner time-scale. By using \eqref{eq:prprop41} and the fact $n_t=m^tn_0$, we have
\begin{equation} \label{eq:prprop49}
\begin{split}
\sum_{i=0}^t K_i = \left\lfloor \frac{D m^t n_0}{w} \right\rfloor, ~\forall~t=0,1,\ldots.
\end{split}
\end{equation}
By substituting \eqref{eq:prprop49} into \eqref{eq:prprop415} and using the fact $k_t/k \le K_t/\sum_{i=0}^t K_t$, we have
\begin{align} \label{eq:prprop410}
&\min_{0 \le i \le {k-1}} \mathbb{E}\left[ \| \nabla F(\bbx_{i})\|_2^2 \right] \le \frac{D}{k} \big(1+1+ \sum_{i=1}^{t}\frac{\frac{D m^{i} n_0}{w} -\frac{D m^{i-1} n_0}{w}+1}{\frac{D m^{i} n_0}{w}-1} \big) = \frac{D}{k} \big(1+1+ \sum_{i=1}^{t}\frac{m-1}{m} + \sum_{i=1}^{t}\frac{2-\frac{1}{m}}{\frac{D m^{i} n_0}{w}-1} \big).
\end{align}
The third term $\{ (2-1/m)/\left((Dm^in_0)/w-1\right) \}_{i=1}^\infty$ in the bound of \eqref{eq:prprop410} is a convergent series such that
\begin{equation} \label{eq:prprop411}
\begin{split}
\sum_{i=1}^{t}\frac{2-\frac{1}{m}}{\frac{D m^{i} n_0}{w}-1} \le \sum_{i=1}^{\infty} \frac{2-\frac{1}{m}}{\frac{D m^{i} n_0}{w}-1} \le c
\end{split}
\end{equation}
with constant $c$. By substituting \eqref{eq:prprop411} into \eqref{eq:prprop410}, we get
\begin{equation} \label{eq:prprop412}
\begin{split}
&\min_{0 \le i \le {k-1}} \mathbb{E}\left[ \| \nabla\! F(\bbx_{i})\|_2^2 \right]\! \le \frac{(2+c)D}{k} + \frac{m-1}{m} \frac{Dt}{k}.
\end{split}
\end{equation}
Now according to \eqref{eq:prprop49}, we have $t \le \log_m \frac{(k+1)w}{n_0D}$. By substituting this result into \eqref{eq:prprop412}, we get
\begin{equation} \label{eq:prprop414}
\begin{split}
&\min_{0 \le i \le {k-1}} \mathbb{E}\left[ \| \nabla F(\bbx_{i})\|_2^2 \right] \!\le\! \frac{D \log_m \frac{kw+w}{n_0D}+(2\!+\!c)D}{k}.
\end{split}
\end{equation}
Therefore, the convergence rate is approximately $\ccalO(\log k / k)$.
\end{proof}

\bibliographystyle{IEEEtran}
\bibliography{myIEEEabrv,biblioOp}




\end{document}